\RequirePackage{fix-cm}
\documentclass[smallextended]{svjour3}       % onecolumn (second format)
\smartqed  % flush right qed marks, e.g. at end of proof

\usepackage{times}
\usepackage{amsmath}
\usepackage{amssymb}

\usepackage[pdftex]{graphicx}
\usepackage{wrapfig}
\usepackage{epstopdf}
\usepackage[american]{babel}
\usepackage{url}
\usepackage{color}
\usepackage{xspace}
\usepackage{float}
\usepackage{tabularx}
\usepackage{multirow}
\usepackage{alltt}
\usepackage{multicol}
\usepackage{blindtext}
\usepackage{scrextend}
\addtokomafont{labelinglabel}{\sffamily}

\usepackage{subfig}

\usepackage{algorithmic}
\usepackage{algorithm} % For counting chapters
\algsetup{linenosize=\scriptsize}
\xspaceaddexceptions{=}
\xspaceaddexceptions{\}}
\xspaceaddexceptions{\in}

\begin{document}

\title{On Cycling Risk and Discomfort: Urban Safety Mapping and Bike Route Recommendations}


\author{David Castells-Graells \and
        Christopher Salahub \and %etc.
        Evangelos Pournaras
}


\institute{David Castells-Graells \and Christopher Salahub \and Evangelos Pournaras \at
			  Professorship of Computational Social Science, ETH Zurich\\
              Stampfenbackstrasse 48, 8092, Zurich, Switzerland \\
              Tel.: +41446320458 \\
              \email{\{davidcas,csalahub\}@student.ethz.ch}, epournaras@ethz.ch           %  \\
}

\date{Received: date / Accepted: date}

\maketitle

\begin{abstract}

\vspace{-0.5cm}
Bike usage in Smart Cities becomes paramount for sustainable urban development. Cycling provides tremendous opportunities for a more healthy lifestyle, lower energy consumption and carbon emissions as well as reduction of traffic jams. While the number of cyclists increase along with the expansion of bike sharing initiatives and infrastructures, the number of bike accidents rises drastically threatening to jeopardize the bike urban movement. This paper studies cycling risk and discomfort using a diverse spectrum of data sources about geolocated bike accidents and their severity. Empirical continuous spatial risk estimations are calculated via kernel density contours that map safety in a case study of Zurich city. The role of weather, time, accident type and severity are illustrated. Given the predominance of self-caused accidents, an open-source software artifact for personalized route recommendations is introduced. The software is also used to collect open baseline route data that are compared with alternative ones that minimize risk or discomfort. These contributions can provide invaluable insights for urban planners to improve infrastructure. They can also improve the risk awareness of existing cyclists' as well as support new cyclists, such as tourists, to safely explore a new urban environment by bike. 

\keywords{cycling \and bike \and accident \and severity \and weather \and Zurich \and risk \and safety \and route \and recommendation \and Smart City \and kernel density}\vspace{-0.5cm}
\end{abstract}

\section{Introduction}\label{sec:introduction}

The use of bikes is transforming urban environments to Smart Cities~\cite{smartcity,Pournaras2018} with the capacity to meet challenging sustainable development goals: cycling supports a more healthy lifestyle, it can decrease energy consumption and carbon emissions in urban centers, reduce traffic jams, limit the need to extensive car parking infrastructures and instead unfold opportunities for building parks, greenery and recreation areas. Bikes also introduce a new experience for tourists to explore a city. Bike sharing infrastructures massively expanding in urban centers show a climax of this transformation.

However, there is evidence that bike accidents are rising as well: While cycling traffic has increased by 35\% since 2013 in Z\"urich, reported bike accidents have increased\footnote{Available at https://www.limmattalerzeitung.ch/limmattal/zuerich/autofahren-wird-sicherer-auf-zwei-fraedern-ist-das-gegenteil-der-fall-132314683 (last accessed: May 2019).} by 60\%. Cyclists' accidents are 5 to 6 times higher per travelled kilometer than that of car occupants~\cite{Kaygisiz2017} in Norway. It is predicted that the accident rate of cyclists is almsot 20 times higher than that of car occupants, when unrecorded bike accidents are considered~\cite{Kaygisiz2017}. Almost half million cyclists die every year in traffic accidents~\cite{World2015}. New insights about why bike accidents happen and how to decrease them are imperative for the adoption of bike as a predominant transport mean in sustainable Smart Cities. 

This paper introduces a data-driven estimation of cycling risk and discomfort, which is used to map safety in the city of Z\"urich. The proposed estimation model is based on kernel density contours that can provide a \emph{continuous estimation of risk on the traffic network}. The severity of the accidents, their causes, the role of the weather/seasonality as well as the day and time that accidents occur are expensively studied using a diverse spectrum of data sources from public authorities, health insurance policies, OpenStreetMap traces as well as typical routes collected from Z\"urich cyclists. The predominance of self-caused accidents indicate the potential to improve bike safety via personalized route recommendations generated by an open-source software artifact with which users can balance cycling safety and comfort.

Earlier work on bike safety studies environmental and demographic factors related to cycling safety, e.g. age, gender, daylight conditions and use of helmet~\cite{RODGERS1995215,RIVARA1997}. Findings focus on the limited protection by helmets or the necessity that a child is developmentally ready for cycling. Data are mainly analyzed at an aggregate national level in the United States and location-specific risks are not taken into account. 

Other work focuses on the design of risk metrics to measure safety such as the concept of \emph{exposure} that accounts for the cycling distance and time spent before an accidents occurs~\cite{BIKEEXPOSURE}. The measure of exposure requires the choice of specific areas and times for modeling, in contrast to the approach introduced in this paper that can generalize the risk estimate to a continuous geographic spectrum. Moreover, the concept of exposure conveys the potential of an accident, while the risk estimate of this paper is explicitly based on actual accident data reported to official authorities. The relation between risk and exposure as well as other methodologies to model and measure cycling safety are reviewed extensively in a recent report by U.S. Department of Transportation~\cite{USDEPTRANREP}, without though providing any quantitative data analysis as performed in this paper.  

Other related work focuses on assessing the cycling routes safety in Berlin by counting the number of hot spots and dangerous intersections that a route contains~\cite{CyclingSafetyRouter}. No continues risk estimation is provided and all such route features are treated and counted equally. Spatio-temporal analysis frameworks of bike accidents are earlier introduced based on network-based kernel density estimation~\cite{Bil2013,Kaygisiz2017}. For instance, applicability in the city center of Vienna, Austria, reveals that bike accident hot spots vary in space according to season, light, and precipitation conditions, while these hot spots cluster by intersections and bus/tram/subway/bike stations~\cite{Kaygisiz2017}. Although the scope of this work is the closest to this paper, it neither covers route discomfort nor route recommendations. 

The contributions of this paper are outlined as follows: (i) A continuous spatial risk and route discomfort estimation model for mapping cyclists' urban safety. (ii) A personalized route recommendation system that balances cycling risk and discomfort. (iii) The design of a novel data analytics pipeline to map cycling risk that combines processes receiving as input a diverse spectrum of data sources. (iv) Findings about the number of accidents, their severity and their causes in the area of Z\"urich, the influence of weather/seasonality as well as daily/weekly accident patterns. (v) An open-source software artifact\footnote{Available at https://github.com/Salahub/BikeRouteRecommender (last access: May 2019)} for the interactive collection of bike route data as well as the computation of personalized route recommendations. (vi) An open dataset~\cite{dataset2019} of cyclists' bike routes in the area of Zurich that can be used as baselines for multi-objective bike route optimization.

This paper is organized as follows: Section~\ref{sec:risk-model} and~\ref{sec:discomfort-model} introduce the spatial risk and route discomfort estimation models respectively. Section~\ref{sec:recommendations} introduces the concept of personalized route recommendations that balance safety and comfort. Section~\ref{sec:methodology} illustrates the experimental methodology for the evaluation of the risk estimation model as well as a software artifact for data collection and bike route recommendations. Section~\ref{sec:evaluation} illustrates the findings of the performed data analysis. Finally, Section~\ref{sec:conclusion} concludes this paper and outlines future work. 

%\vspace{-0.5cm}

\section{Spatial Risk Estimation Model}\label{sec:risk-model}

This section introduces a general-purpose data-driven model for the spatial estimation of transport risk using geolocated traffic and accident data. Table~\ref{eq:math} illustrates the main mathematical symbols and their data applicability. Risk is the conditional probability of involvement in a traffic accident $A$ given the use of a particular transit method $T$. It is represented by a conditional probability density $A|T$ as follows: 

\begin{equation}
  f_{A|T}(\textbf{x}) = \frac{f_{A,T}(\textbf{x})}{f_T(\textbf{x})},
  \label{eq:ConditionalDens}
\end{equation}

\noindent where $f_A(\textbf{x})$ and $f_T(\textbf{x})$ are the respective marginal densities and $f_{A,T}(\textbf{x})$ is the joint density. $f_{A,T}(\textbf{x})$ can be viewed as a normalization or regularization of $f_A(\textbf{x})$, which accounts for the traffic level at a certain location. In practice, the estimation of the conditional density, $f_{A|T}(\textbf{x})$, is feasible with geolocated accident and transit data for the transit method $T$, representing samples from the joint distribution $f_{A,T}(\textbf{x})$ and the marginal density $f_T(\textbf{x})$ respectively. Given that the individual accident coordinates $\{\textbf{x}_i\}_{i=1,\dots,n}$ are discrete points, a continuous and non-parametric density estimate can be calculated using kernel density estimation\footnote{Rather than estimating a global density model, estimation consists of averaging local density models, i.e. kernels. The final model is not a member of a parametrized distribution family, though the local density estimates may be.} (KDE)~\cite{KDE}. 

\begin{center}
\begin{table}[!htb]
\centering
\caption{Notation of the spatial risk estimation model and its applicability.}\label{eq:math}
\begin{tabular}{|c|c|c|}
\hline
\textbf{Notation} & \textbf{Meaning} & \textbf{Applicability} \\
\hline
$\textbf{x} \in \Re^m$ & Estimation coordinate & Z\"urich street network grid \\
$\textbf{x}_i \in \Re^m$ & Observation coordinate & Accident and traffic data \\
$m$ & Dimension of coordinates & 2 \\
$n$ & Number of observations & 1305 accidents, 242801 GPS trace points \\
$h$ & Scalar bandwidth parameter & 0.003 \\
$K_h(\cdot)$ & $h$-parameterized kernel & $\frac{1}{\sqrt{2 \pi h}^m} e^{\frac{\textbf{x}^T \textbf{x}}{2h}}$ \\
$\textbf{x}^T$ & Transpose operation & - \\
$\hat{f}$ & Estimator of $f$ & - \\
$A_s$ & $s^{th}$ partition of the accident data & Accidents with severity $s$ \\
$S$ & Number of partitions & 3 \\
\hline
\end{tabular}
\end{table}
\end{center}

\vspace{-0.5cm}

Given $n$ points in $d$ dimensions and a local density function, or kernel, the equation for
estimating the density at a point $\textbf{x}$ using the set $\{\textbf{x}_i\}_{i=1,\dots,n}$ for a given kernel $K$ parameterized locally by $h$ is given as follows:

\begin{equation}
\hat{f}(\textbf{x}) = \frac{1}{n} \sum_{i=1}^n K_h \left ( \textbf{x} - \textbf{x}_i \right ),\label{eq:KernelDen}
\end{equation}

\noindent where $K_h$ is the kernel function normalized such that integration over its local support is one, e.g. the Gaussian density:

\begin{equation}
K_h(\textbf{x}) = \frac{1}{(\sqrt{2 \pi h})^m} e^{\frac{\textbf{x}^T \textbf{x}}{2h}}.\label{eq:Gaussian-kernel}
\end{equation}

\noindent An isotropic zero-centered Gaussian kernel ensures that the local density estimation does not preferentially estimate densities in any transit direction. This is applicable if the orientation of all streets within a region is not known a priori. However, transit points can be spuriously related using such a kernel, e.g. accidents that occur on parallel but disconnected streets. 

The kernel density estimation assumes a homogeneous, i.e. linear, influence of the geolocated accident data that estimate the risk. However, additional meta-information can be employed to polarize the risk at certain accident locations, for instance, and assign a higher weight to accidents that result in a death than those resulting only in minor injuries. Assume that $\{\textbf{x}_i\}_{i=1,\dots,n}$ can be partitioned into $S$ meaningful subsets indexed by $s$. Classifying the geolocated accident data by severity gives subsets $A_s$ of size $n_s = |A_s|$ labeled by meta-information about accident severity. The kernel density is reestimated in this case according to Corollary~\ref{eq:heterogeneous}.

\begin{corollary}\label{eq:heterogeneous}
The kernel density estimation of geolocated accident data classified in $S$ subsets each with size $n_s$ is calculated as follows:

\begin{equation}
\hat{f}(\textbf{x}) = \sum_{s=1}^S \frac{n_s}{n} \hat{f}_{A_s}(\textbf{x}),
\label{eq:KernDenEqReweight}
\end{equation}

\noindent where $\hat{f}_{A_s}$ is the kernel density estimated with the data of subset $A_s$.

\end{corollary}

\begin{proof}

The sum over all $n$ elements can be expanded as the sum over all subsets and subset elements $\textbf{x}_{sj}$ as follows:

\begin{equation}
\hat{f}(\textbf{x})=\frac{1}{n}\sum_{i=1}^n K_h \left (\textbf{x}-\textbf{x}_i \right)=\frac{1}{n} \sum_{s=1}^S \sum_{j=1}^{n_s} K_h \left (\textbf{x}-\textbf{x}_{sj} \right).
\label{eq:proof-step-a}
\end{equation}

\noindent The summands can be multiplied with $1 = \frac{n_s}{n_s}$. However, $n_s$ is constant within the second summand over $A_s$ and as such the numerator is moved outside the inner sum. Similarly,
$\frac{1}{n}$ can be moved inside the first summand. Based on Equation~\ref{eq:Gaussian-kernel} the kernel density estimation of partition $s$ is derived: 

\begin{equation}
\frac{1}{n} \sum_{s=1}^S \sum_{j=1}^{n_s} K_h \left (\textbf{x}-\textbf{x}_{sj} \right)=\sum_{s=1}^S \frac{n_s}{n}\sum_{j=1}^{n_s} \frac{1}{n_s} K_h \left (\textbf{x}-\textbf{x}_{sj}\right)=\sum_{s=1}^S\frac{n_s}{n} \hat{f}_{A_s}(\textbf{x}).
\label{eq:proof-step-b}
\end{equation}
\end{proof}

%The above Corollary suggests that the kernel density can be independently evaluated on any data partition while it can be recombined with the correct weights to obtain the original density.

Partitions can be reweighted to reflect the consequence of an accident type, for instance insurance compensations for accidents of the $S$ partitions, or the relative importance of
the partitions $S$ to a policy directive. The following equation:

\begin{equation}
f_{\mathtt{R}}(\textbf{x}) = \sum_{s=1}^S a_s \hat{f}_{A_s}(\textbf{x}),
\label{eq:reweighting}
\end{equation}

\noindent is such a generalization, where $\frac{n_s}{n}$ corresponds to weights $a_s$ such that $\sum_{s = 1}^S a_s = 1$. 

%\vspace{-0.5cm}

\section{Route Discomfort Estimation Model}\label{sec:discomfort-model}

Along with the risk estimation, discomfort estimation can be used to assign a respective weight to each edge of the street network. The \emph{discomfort} of a bike route is defined by the level of physical effort required to traverse the route by bike in terms of length and grade, i.e. slope.   Distinguishing the contributions of the two variables with a closed expression is not straightforward. A vast amount of data and domain experience is required for designing a realistic model. The \textit{IBP index}\footnote{Available at https://www.ibpindex.com (last accessed: May 2019).} is introduced in the context of bike route discomfort. It is generated by an algorithm that analyzes the difficulty of a mountain or bike route and it is backed by years of iteration between measurement and adjustment. This index is currently used by many associations and guides, for instance, the French Federation of Hiking. Although the IBP index is not based on a single closed expression given a series of corrections and fits accounting for multiple other parameters, it can be used to calibrate one such expression.

To make use of the the IBP index, different synthetic tracks of constant grade and specific length are generated and formatted in GPX files required for the input to the IBP index application. A linear dependency with length and an exponential one with grade is inferred after inspecting the outputs and the fitting curves. The discomfort expression is given in the following form:

\begin{equation}\label{comfortFnc}
f(d,x) = \begin{cases}
d\cdot(2\,\exp(15\,x)-1) & \textrm{if } x \geq -0.025 \\
f(d,-0.025) & \text{otherwise,}
\end{cases}
\end{equation}    
          
\noindent where $d$ is the length of a street or route and $x$ is the average grade: $-1\leq x \leq 1$. The expression is set constant below a -2.5\% grade, as grade values lower than this threshold make no difference in effort for a cyclist. This limit is in line with the results from the IBP index calculator. Note that at zero grade, $f(d,0)\propto d$.

\vspace{-0.5cm}

\section{Personalized Route Reccomendation Based on Risk and Discomfort}\label{sec:recommendations}

Personalized route recommendations between a departure and destination point are calculated using Breadth First Search (BFS) over the street network with assigned weights on the streets. These weights are a cyclist's determined combination of the risk and discomfort estimates as follows:

\begin{equation}\label{eq:edge-weights}
w=\alpha\,w_{\mathtt{r}}+(1-\alpha)\,w_{\mathtt{d}},
\end{equation} 

\noindent where $w_{\mathtt{r}}$, $w_{\mathtt{d}}$ are the risk and discomfort weights respectively. An $\alpha=0$ prioritizes routes with minimal discomfort, while $\alpha = 1$ prioritizes routes with minimal risk. 

%\vspace{-0.5cm}

\section{Data Science and Experimental Methodology}\label{sec:methodology}

This section introduces a realization of the proposed risk estimation model. It also introduces a software artifact for bike route recommendations. Figure~\ref{DataFlowFig} outlines the bike riding risk assessment in Z\"urich city. A data science pipeline is designed that consists of the following stages: \textbf{Stage 1}: Extraction and categorization of bike accident data from the Swiss GeoAdmin\footnote{Available at https://map.geo.admin.ch/?topic=vu\&lang=de\&bgLayer=ch.swisstopo.pixelkarte-grau\&layers=ch.astra.unfaelle-personenschaeden\_alle\&layers\_timestamp=\&catalogNodes=1318 (last access: May 2019).} API\footnote{Available at http://api.geo.admin.ch (last access: May 2019)} to obtain $\{\textbf{x}_i\}_{i=1,\-\dots,\-n} = \{A_s\}_{s = 1,\-\dots,\-S}$. \textbf{Stage 2}: Extraction and processing of GPS trip traces from the OpenStreetMaps (OSM) API\footnote{Available at https://www.openstreetmap.org/export (last access: May 2019)}. \textbf{Stage 3}: Extraction and processing of the Z\"urich street network from the Swiss GeoAdmin~\footnotemark[5] API\footnotemark[6]. \textbf{Stage 4}: The kernel density estimation on traces and labeled accidents to calculate $f_T(\textbf{x})$ and $f_{A_s,T}(\textbf{x})\hspace{0.1cm} \forall s$. \textbf{Stage 5}: Calculation of $\hat{f}_{A_s|T}(\textbf{x}) \forall s$ according to Equation~\ref{eq:ConditionalDens} by taking the ratio of $\hat{f}_T(\textbf{x})$ and $\hat{f}_{A_s,T}(\textbf{x}) \forall s$. \textbf{Stage 6}: Application of insurance recompensation data\footnote{Available at https://www.bj.admin.ch/dam/data/bj/gesellschaft/opferhilfe/hilfsmittel/leitf-genugtuung-ohg-d.pdf (last access: May 2019).} from the Swiss Federal Office of Justice to obtain $f_{\mathtt{R}}(\textbf{x})$ from $\hat{f}_{A_s|T}(\textbf{x})$ according to Equation~\ref{eq:reweighting}. \textbf{Stage 7}: Interpolation of $f_{\mathtt{R}}(\textbf{x})$ onto the processed street network.

\begin{figure}[!htb]
\centering
\includegraphics[width=0.9\linewidth]{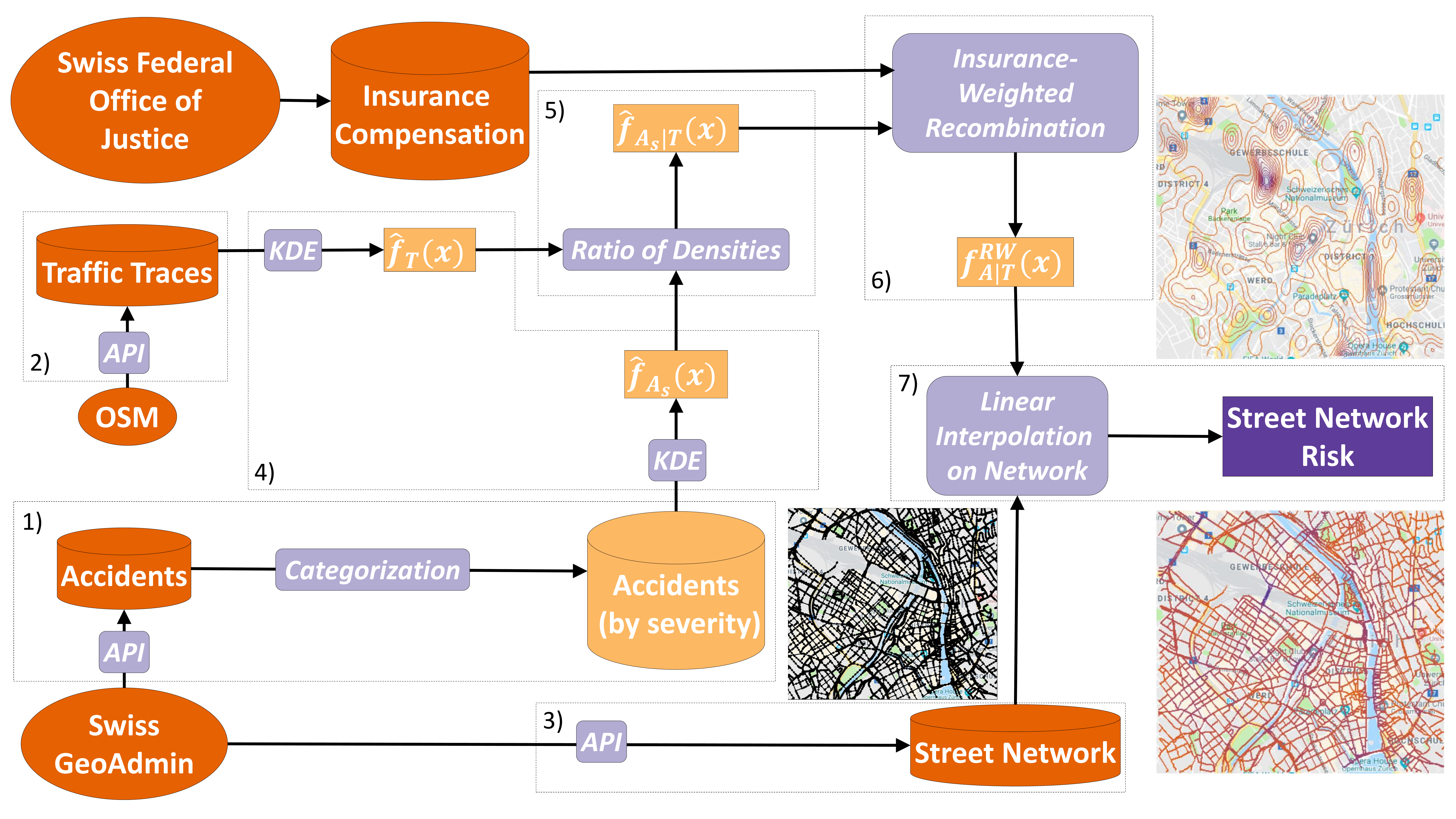}
\caption{An overview of the data science pipeline for bike riding risk assessment.}\label{DataFlowFig}
\end{figure}

\subsection{Stage 1: accident data extraction}\label{subsec:accidents}

The Swiss Confederation\footnotemark[5] web portal has an interactive map of Switzerland with several spatial layers of publicly available data. One of these layers compiles and displays accidents involving bikes between 2011 and 2017. The data are collected by the Swiss Federal Roads Office from electronic police reports\footnote{Available at https://www.astra.admin.ch/astra/de/home/dokumentation/unfalldaten/grundlagen/\-prozess.html (last access: May 2019)}. Together with the localization of the accidents, this layer provides an accident specification that includes the date and time, severity, cause, and street type. These features are visualized on the map, which also serves as the basis of an API service for batch data extraction\footnotemark[6].

The API service `identify' is used for data extraction. It generates and returns a list of at most 200 elements from a layer satisfying a geographic geometry specified using ESRI syntax\footnote{This data is defined using the Swiss projection coordinates used by the Swiss GeoAdmin system: https://www.swisstopo.admin.ch/en/maps-data-online/calculation-services.html (last access: May 2019). All values are converted to the more common WGS84 system once extracted using the approximation.}. For simplicity and scalability, the geometry used in this investigation is a bounding box specified by its horizontal and vertical extents as shown in Figure~\ref{Plotgridalgo}a. It is delimited\footnote{The `identify' service requires the specification of other map features that affect the inclusion tolerance of any bounding box: (i) the \emph{map extent} parameter that is the entirety of Switzerland set with 312250, -77500, 1007750, 457500 and (ii) the \emph{image display} parameter chosen to be 1391, 1070, 96.} by the latitudes $(47.3650,47.3886)$ and longitudes $(8.5141,8.5523)$. This Z\"urich region is chosen for its central location and the density of recorded trips on OpenStreetMaps (OSM), an indication of frequent trips.

%\vspace{-0.5cm}

\begin{figure}[!htb]
\centering
\subfloat[Selected data region.]{\includegraphics[width=0.275\linewidth]{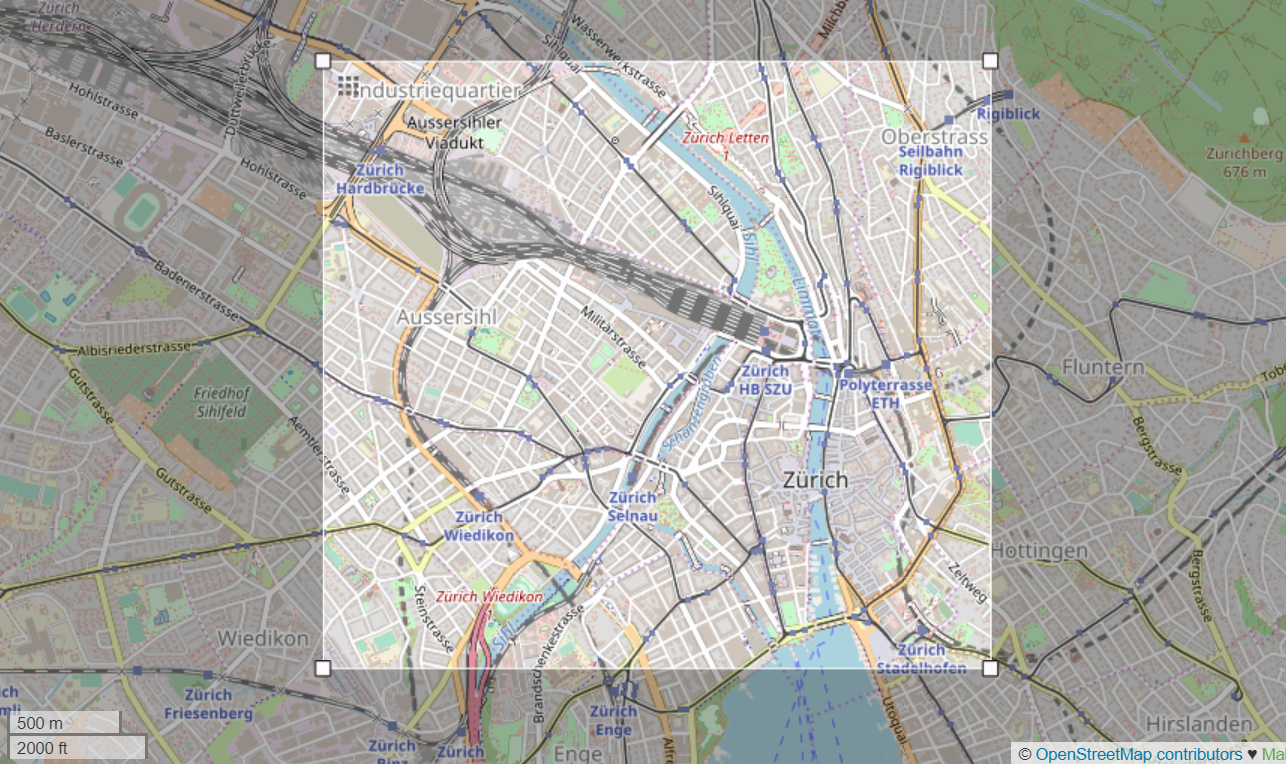}}
\subfloat[Iteration 1]{\includegraphics[width=0.178\linewidth]{grid-01}}
\subfloat[Iteration 2]{\includegraphics[width=0.178\linewidth]{grid-02}}
\subfloat[Iteration 3]{\includegraphics[width=0.178\linewidth]{grid-03}}
\subfloat[Iteration 4]{\includegraphics[width=0.178\linewidth]{grid-04}}
    \caption{The selected region and an example of subdivisions using a threshold of 1.}\label{Plotgridalgo}
\end{figure}

%\vspace{-0.5cm}

To meet the limit of 200 extracted elements, consecutive subdivisions of the bounding box are performed prior to data extraction so that each subdivision in which the final data extraction\footnote{To ensure all data points are obtained including those on the boundaries, the final extraction windows applied to each subdivision are selected with deliberate overlap at the boundaries. This process created duplicate records that are removed using the unique ID associated with each element.} takes place contains 200 elements or fewer. Algorithm~\ref{algo:grid} outlines the subdivision logic
and Figure~\ref{Plotgridalgo}b-\ref{Plotgridalgo}e displays this algorithm schematically using a threshold of 1 rather than 200 for simplicity.

\vspace{-0.5cm}

\begin{algorithm}
  \caption{Grid subdivision logic of the selected region.}\label{algo:grid}
  \begin{algorithmic}
    \STATE set grid subdivision to 1
    \WHILE{number of elements in any subdivision $>$ 200}
    \STATE divide each subdivision into quarters
    \ENDWHILE
    \RETURN Subdivision boundaries
\end{algorithmic}
\end{algorithm}

\vspace{-0.5cm}

The `severity' field from the extracted data is used to categorize the different accidents into \emph{light}, \emph{severe} and \emph{death} injuries.

\subsection{Stage 2: trip data extraction}\label{subsec:trips}

Transit data from users' GPS traces are downloaded in XML format from OpenStreetMaps\footnotemark[7] and treated as a sample of $f_T(\textbf{x})$. A 5\% of the traces are labeled by means of transport, from which all non-bike traces are removed to improve the quality of estimation of $f_T(\textbf{x})$. For unlabeled traces, travel homogeneity between methods of transit is assumed as streets included in the selected data window are primarily multi-use, and so a high volume of general usage implies high bike usage as well. 

This assumption is also motivated by the very limited available data sources that can be used in the scope of this paper. For instance, Open Data Z\"urich\footnote{Available at https://data.stadt-zuerich.ch/
(last access: May 2019).} has more precise bike traffic data, but they lack resolution and cannot be used for kernel density estimation. Although the introduced assumption imposes limitations on the generalization of the performed analysis, a comparison between $\hat{f}_{A,T}(\textbf{x})$ and $\hat{f}_{A|T}(\textbf{x})$ in Figure~\ref{PlotNormalizationComp} shows smooth adjustments rather than drastic changes to the general trends, suggesting that the data pattern is not lost due to the imposed imprecision. 

\subsection{Stage 3: street network data extraction}\label{subsec:network}

The Z\"urich street network required for $f_{A|T}(\textbf{x})$ is extracted from the Swiss Confederation web portal\footnote{Available at https://map.geo.admin.ch/?topic=ech\&lang=en\&bgLayer=ch.swisstopo.pixelkarte-farbe\&catalogNodes=457,458\&layers=ch.swisstopo.swisstlm3d-strassen (last access:May 2019)} as shown in Stage 1. It is stored by the coordinates (latitude, longitude, altitude) of points along the city streets and the street segments linking these points. While the street points are not equidistant, they are always present at intersections. Moreover, the large number of points defining the street network makes computations, such as graph search used for route recommendations, very expensive. The computational load is reduced by keeping the intersection points only, which is roughly 10\% of the total street data elements.

After the data extraction and pre-processing, the extracted data are modeled as a graph structure. Each street network point is assigned to a node and the points are connected by the edges representing street segments. In terms of the accuracy in this process, the source data are already grouped into street segments and the coordinates match at the intersections, with some tolerance in the last decimal digits.

\subsection{Stage 4-7: risk estimation}\label{subsec:riskest}

The kernel density estimation is performed using the \texttt{kde2d} function in the \texttt{stats} package\footnote{Available at https://www.r-project.org (last access: May 2019).} of \texttt{R}. Different choices of $h$ result in very different estimated densities as suggested in Equation~\ref{eq:KernelDen}. High values of $h$ result in coarse and uninformative densities, while low values result in densities that cannot effectively estimate beyond the immediate neighborhood of $\textbf{x}_i$. This is visually presented in the two examples of Figure~\ref{PlotBandwidths}a and~\ref{PlotBandwidths}b. A bandwidth of 0.003 degrees, in WGS84 coordinates, is selected empirically as it provides a good trade-off. The estimated density contours across the entire accident data are shown in Figure~\ref{PlotBandwidths}c.

\vspace{-0.5cm}

\begin{figure}[!htb]
\centering
\subfloat[Bandwidth=0.001]{\includegraphics[width=0.33\linewidth]{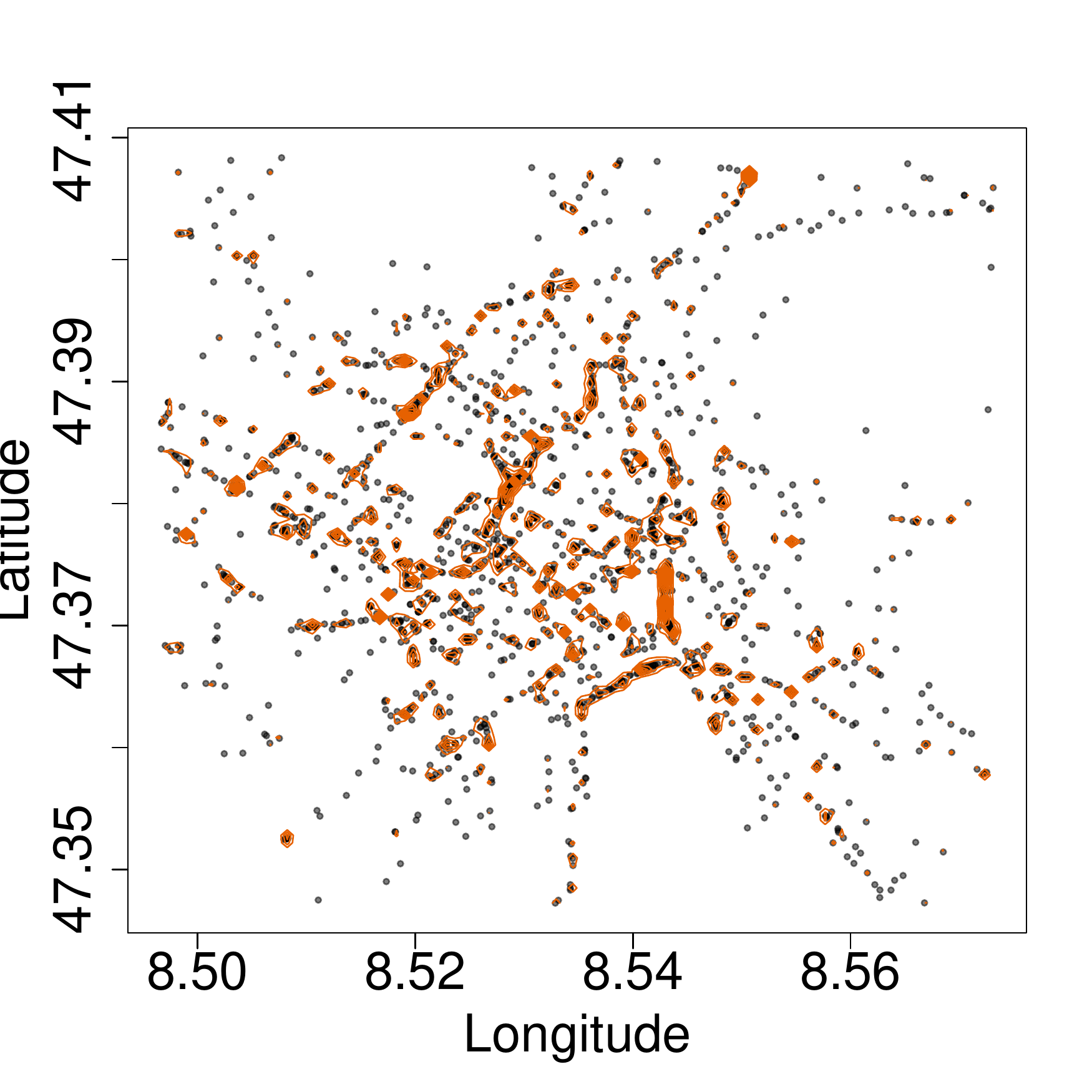}}
\subfloat[Bandwidth=0.01]{\includegraphics[width=0.33\linewidth]{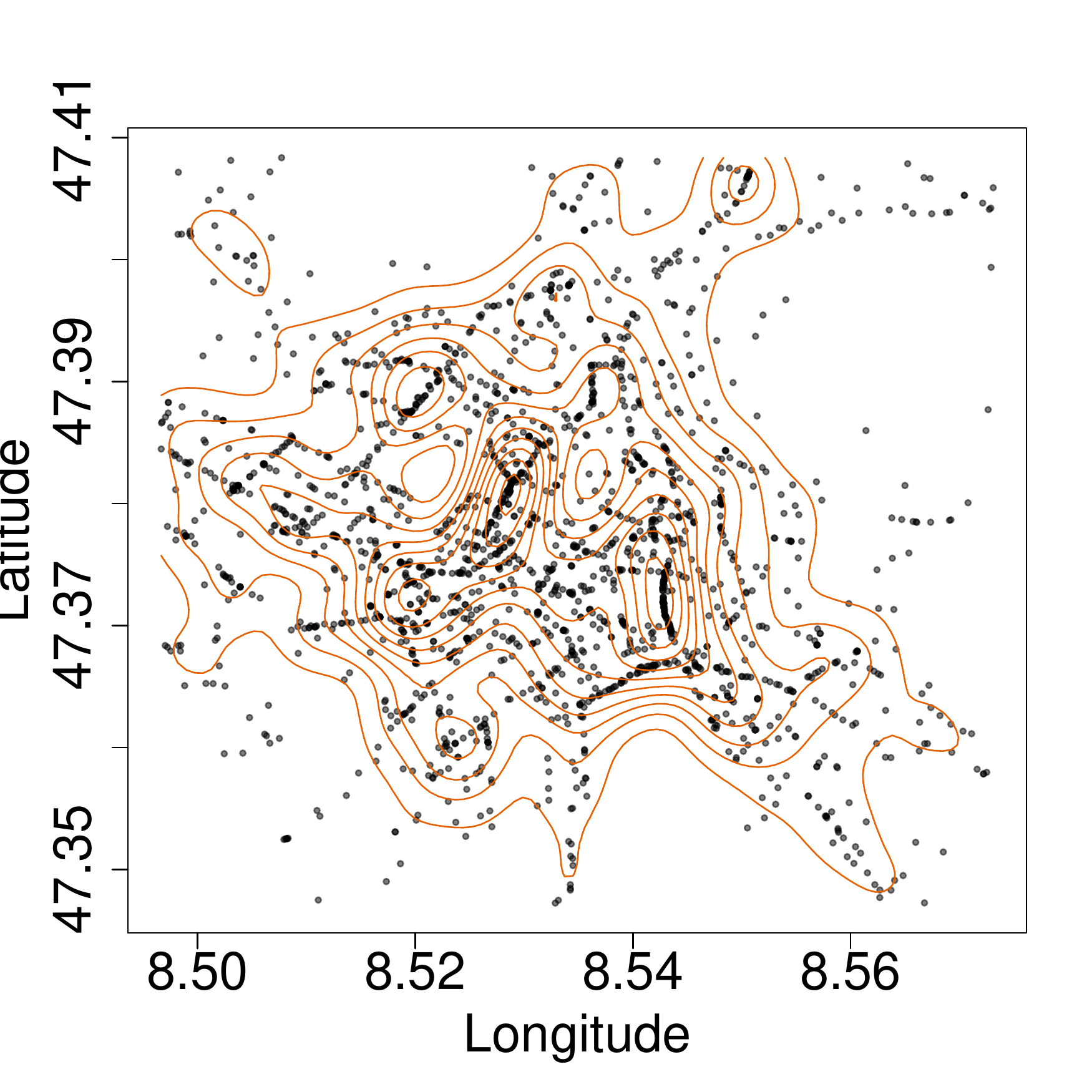}}
\subfloat[Bandwidth=0.003]{\includegraphics[width=0.33\linewidth]{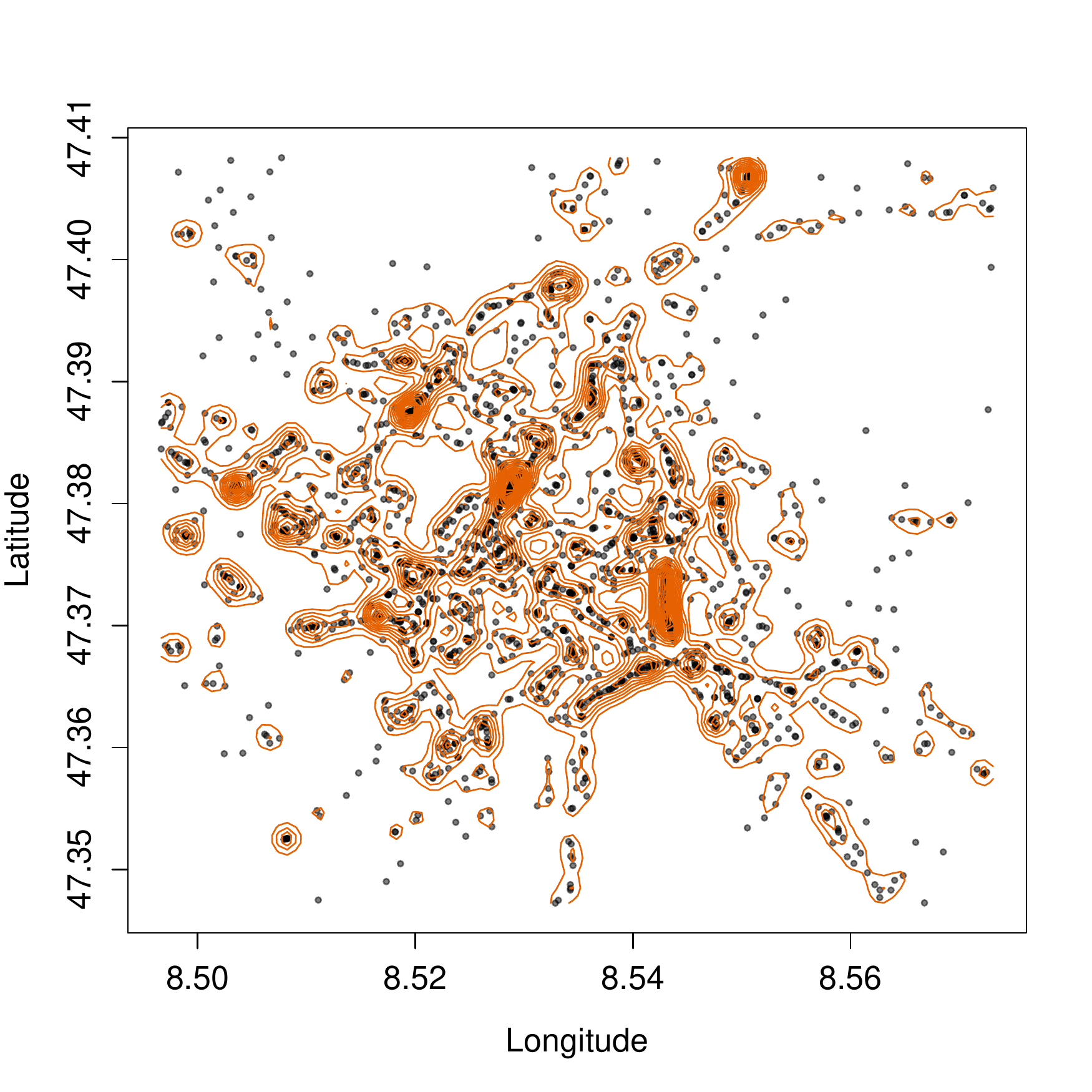}}
\caption{Bicycle accident density contours for different bandwidth choices.}\label{PlotBandwidths}
\end{figure}

\vspace{-0.5cm}

The density $f_{A_s,T}(\textbf{x})$ of bike accidents is calculated for each severity level $s$. Similarly, $f_T(\textbf{x})$ is estimated using the OpenStreetMap traces to finally calculate Equation~\ref{eq:KernDenEqReweight}. As the studied area is not perfectly square, a grid with 560 horizontal and 440 vertical divisions is imposed and estimations are made at the intersections of the grid lines. This asymmetry ensures that the evaluation positions are equispaced in terms of WGS84 coordinates of latitude and longitude. 

To generate $f_{\mathtt{R}}(\textbf{x})$, a relative weighting of 1:6:6 for light injuries, severe injuries, and death respectively is used to recombine the partition densities. These weights are based on the following: (i) Insurance compensation policy data of 5'000:30'000:100'000 CHF for the respective accident severity levels\footnotemark[8]. (ii) The relative factor of 20 for death is scaled down to 6 to prevent over-polarized contours with extreme peaks. 

Low traffic areas, as measured by the volume of OpenStreetMap traffic, may result in contour peaks despite a similar accident rate per trip for these areas, i.e. the ratio of $\hat{f}_{A_s,T}(\textbf{x})/\hat{f}_T(\textbf{x})$ from Equation~\ref{eq:ConditionalDens}. Non-normalized and normalized density contours are compared in Figure~\ref{PlotNormalizationComp}a and~\ref{PlotNormalizationComp}b. The contour peaks of the latter are less extreme than those in Figure \ref{PlotNormalizationComp}a, while the dominant peaks remain the same and distinguishable in both versions, suggesting that normalization has a desirable effect.

\vspace{-0.5cm}

\begin{figure}[!htb]
\centering
\subfloat[Non-normalized contours.]{\includegraphics[width=0.257\linewidth]{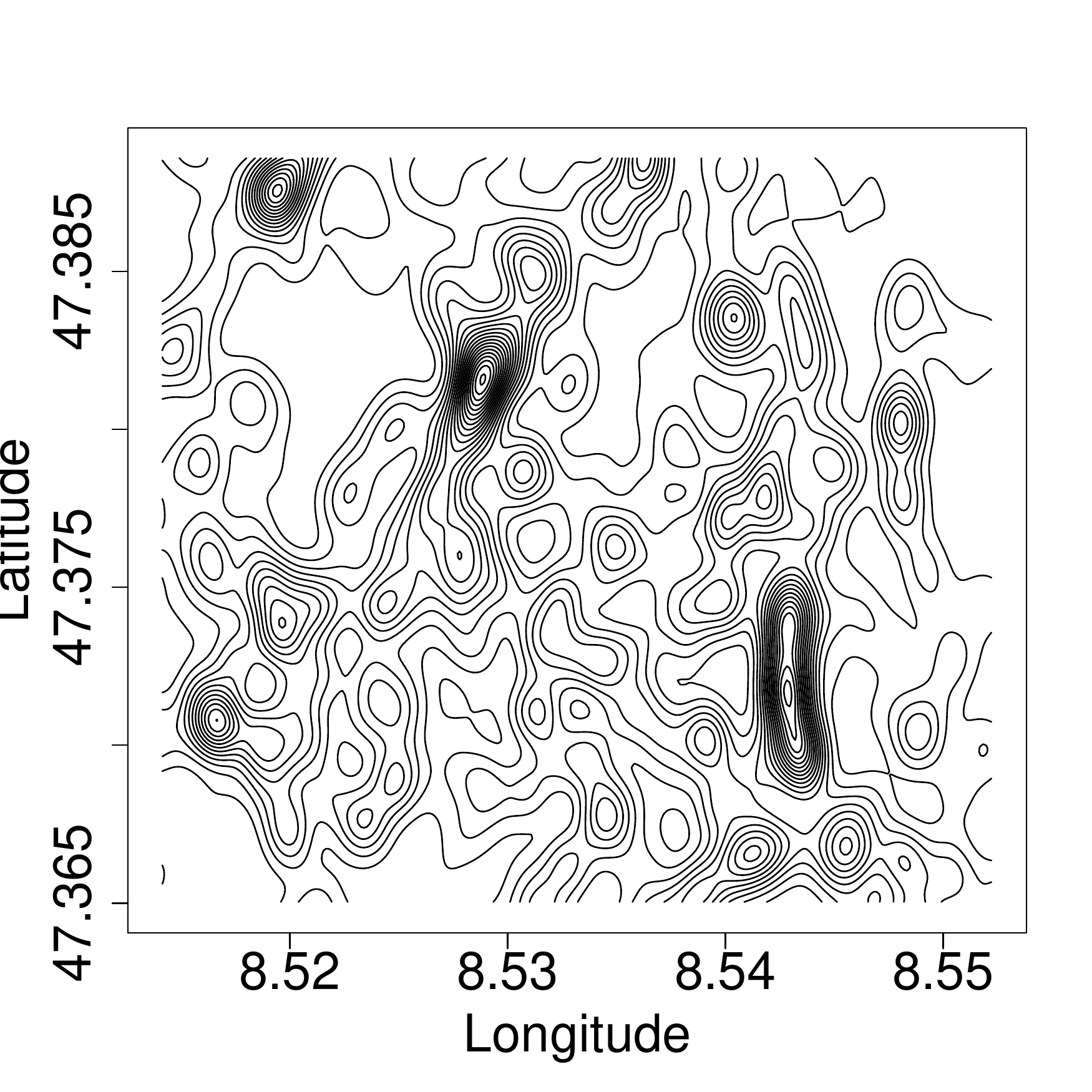}}
\subfloat[Normalized contours.]{\includegraphics[width=0.257\linewidth]{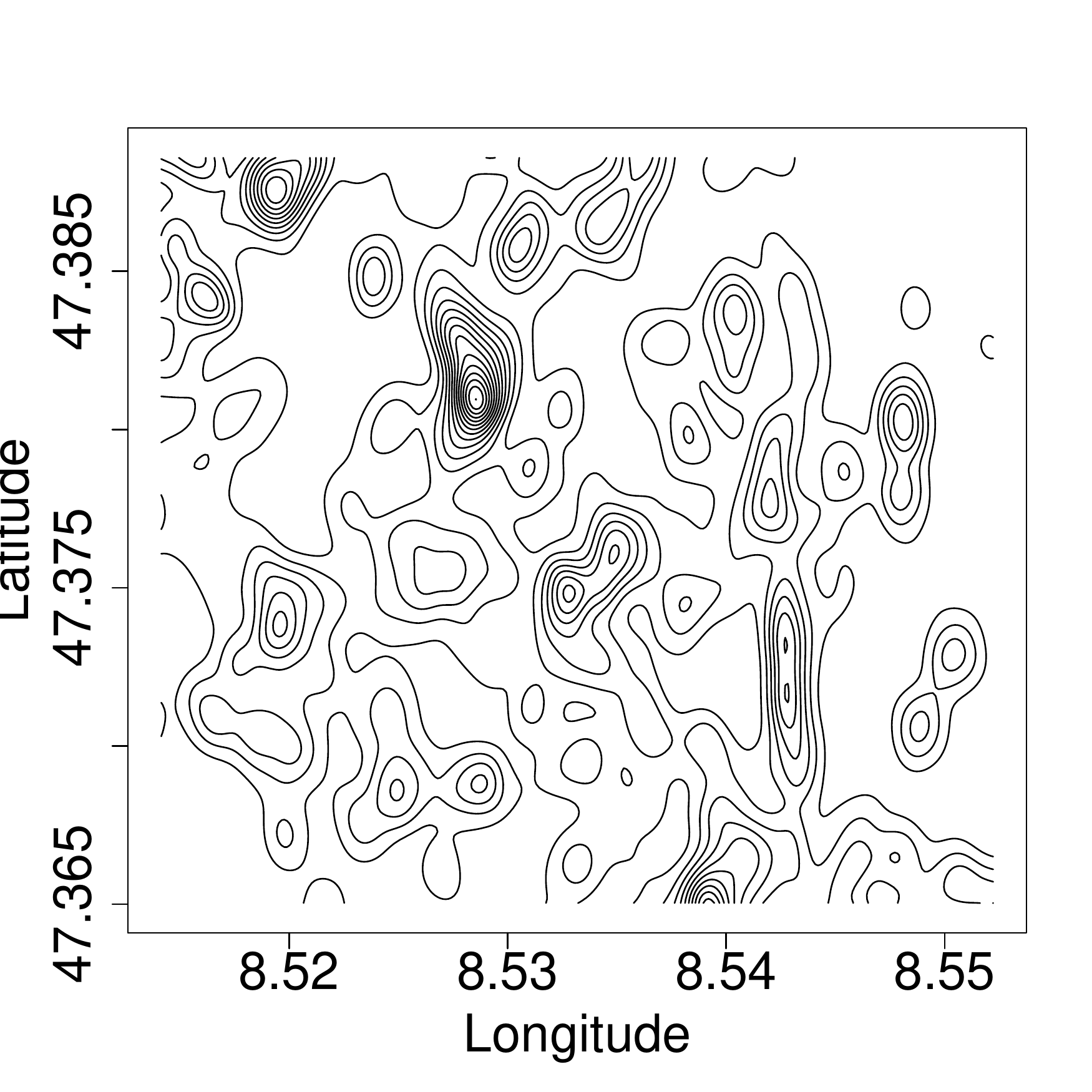}}
\subfloat[Non-scaled interpolation.]{\includegraphics[width=0.256\linewidth]{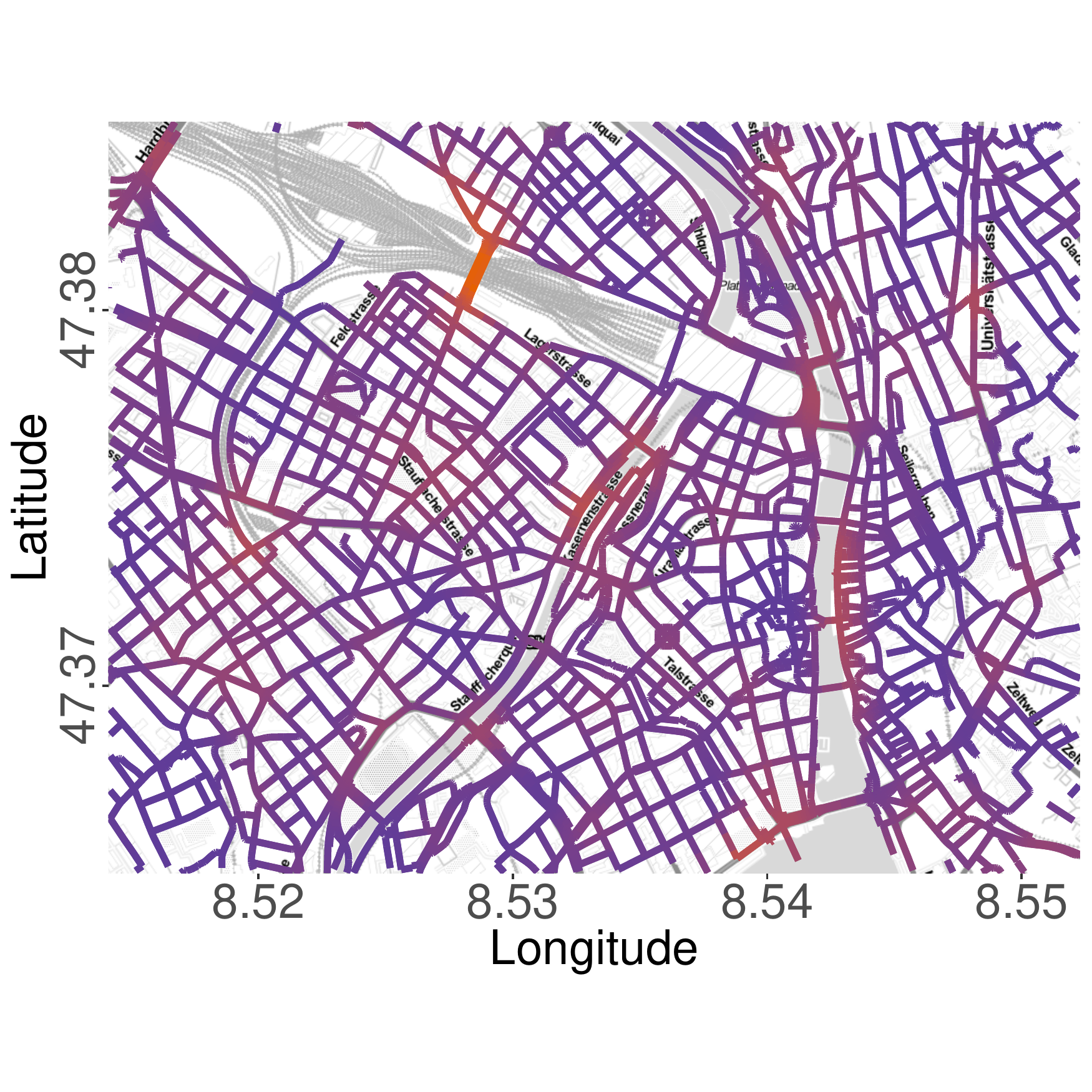}}
\subfloat[Scaled interpolation.]{\includegraphics[width=0.256\linewidth]{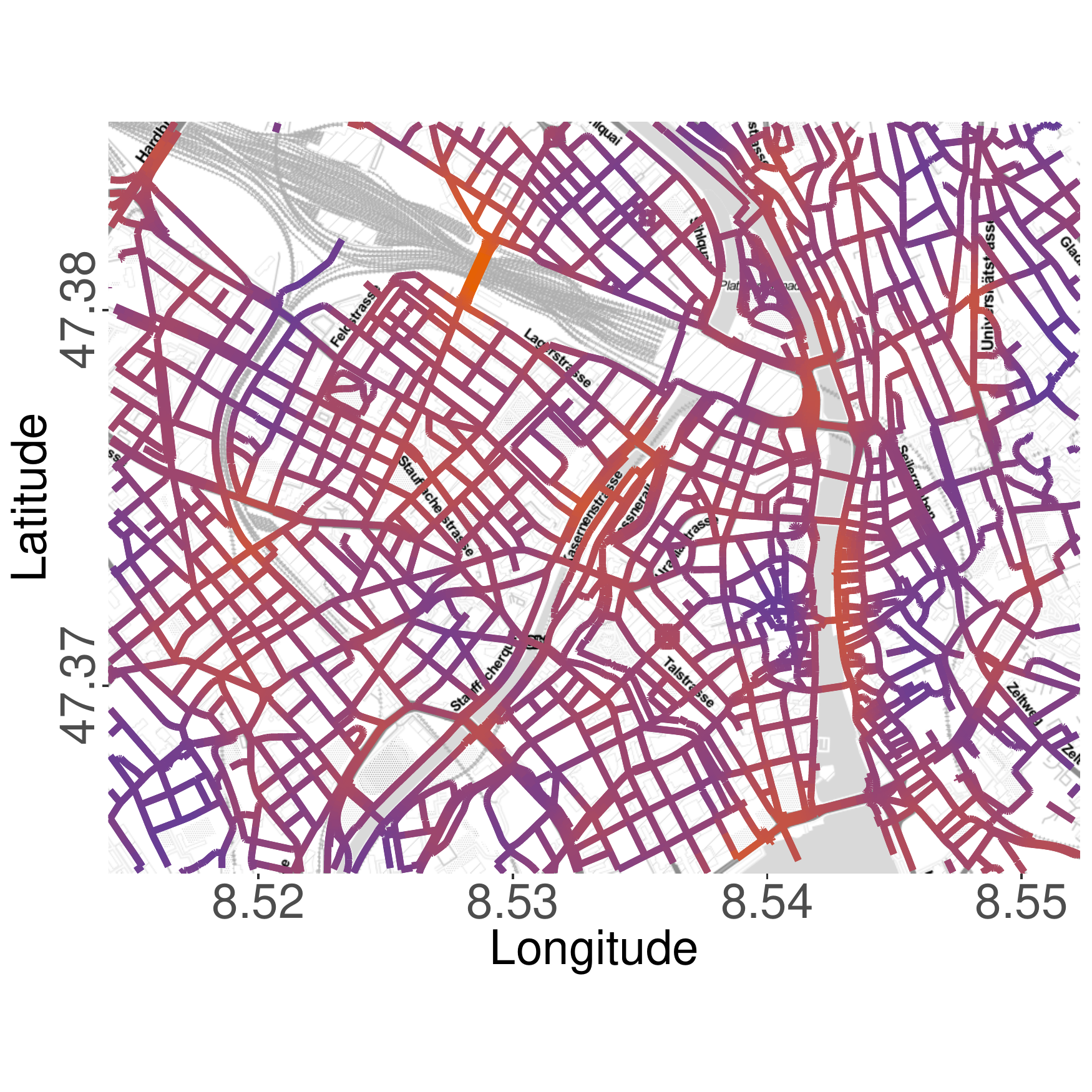}}
\caption{Density contours and interpolated network risks in Z\"urich. Orange hue denotes higher risk.}\label{PlotNormalizationComp}
\end{figure}

\vspace{-0.5cm}

Moreover, note that the estimation window in Figure \ref{PlotBandwidths} extends beyond the specified studied region. Density estimation has highly variable boundary behavior due to the
abruptly exclusion of points at the window edges. This boundary effect, further exacerbated by using the ratio of the densities estimated over the window, results in spuriously peaked boundary estimates of $f_{A_s|T}(\textbf{x})$. An extended window is introduced to estimate the densities restricted back and normalized to the studied region. 

At the final stage, $f_{\mathtt{R}}(\textbf{x})$ is mapped to the street network using simple linear interpolation. The resulting normalized risk is plotted on a map of Z\"urich using the \texttt{ggmap}~\cite{ggmap} and \texttt{ggplot2}~\cite{ggplot2} packages in \texttt{R}. The interpolated risks on the street network are displayed in Figure~\ref{PlotNormalizationComp}c.

Immediately apparent is the relatively high risk in two vibrantly orange areas near Hardbr\"ucke\footnote{Available at https://en.wikipedia.org/wiki/Hardbrücke (last access: May 2019).} and Langstrasse\footnote{Available at https://en.wikipedia.org/wiki/Langstrasse (last access: May 2019).}. These areas are, by a wide margin, the most dangerous in Z\"urich and the magnitude of their risk makes visual risk inspection throughout the rest of Z\"urich challenging. A Box-Cox power transformation~\cite{BoxCox} with an exponent of $\frac{1}{2}$ is applied to the data as shown in Figure~\ref{PlotNormalizationComp}d. The variation in risk is more visually apparent and so it is easier to distinguish higher and lower risk areas.

The risk estimation method illustrated in this paper relies on the quality of the reported accident data. However, it is likely that accidents are under-reported to police, especially those that do not result in injuries or property damage. The following reasoning is made about these unreported accidents: (i) As unreported accidents are expected to be of light severity, they are not expected to significantly increase the estimated risk values. Moreover cyclists are more likely to trust data that are directly linked to reported accidents and convey a consequence or threat. (ii) Under the assumption that unreported accidents are homogeneously distributed in the studied area, their influence on the estimated density contours is negligible.

\subsection{A software artifact for personalized bike route recommendations}\label{subsec:artifact}

A software artifact is introduced to calculate route recommendations on a weighted graph extracted from the street network. The departure and destination points are matched to nodes of the graph. The matching is performed by the distance minimization between the points and nodes. The graph is implemented as a Python object and the software executes the BFT algorithm\footnote{For efficiency, the different explored paths have the accumulated weight assigned and are terminated whenever their accumulated weight exceeds the value of a path that has already reached the destination. They are also terminated whenever they reach a node previously visited by any other path, whose weight at that node is lower than the one of the current path. Once the minimal path is found, the list of nodes is returned together with the calculation of the accumulated risk and discomfort weights.} as illustrated in Section~\ref{sec:recommendations}.

A user interface with an interactive map is designed and implemented in the Python \textit{tkinter} library as shown in Figure~\ref{PlotFreqComp}a. It is used to acquire the required user input to generate a personalized recommended route, i.e. the $\alpha$ weight, the departure and destination point. The user can click on the map to determine these points, as well as intermediate route points, whose latitude and longitude coordinates are computed in the background. The recommended route for a given $\alpha$ is displayed on the map. The user input and the calculated values, i.e. total risk, discomfort and the points that form the route, can be exported in a \texttt{.txt} file.

Two evaluation methods of the software artifact and the route recommendations are introduced: (i) 24 typical bike routes by 8 individuals, who cycle regularly in the studied area, are collected via the software artifact. Data collection is performed via the software artifact by clicking the interactive map several times to form a route that can be then exported as illustrated earlier. These routes are referred to as \emph{baseline routes} and they are compared to the route recommendations for the same departure and destination points and different $\alpha$ values. (ii) A number of 2000 random departure and destination points is generated on the map. For each pair, three route recommendations are generated for $\alpha=0$, $\alpha=0.5$ and $\alpha=0.75$.  Based on these generated routes, the utilization of the street network can be compared for each of these three cases. In this way, the overall risk and discomfort estimation that stems from the route recommendations is mapped on the studied region.

%\vspace{-0.5cm}

\section{Experimental Evaluation}\label{sec:evaluation}

This section illustrates a data analysis of the accidents and evaluates the bike route recommendations.

\subsection{Accident analysis}\label{subsec:accanal}

The total number of reported accidents in the specified region and time frame is 1305: 1023 light injuries, 277 severe injuries, and 5 deaths. Because of the low number of death events, they are included in the group of severe injuries. It is likely that circumstances separating a severe injury from a death may have little to do with accident features and more with the constitution of the victim. Figure~\ref{PlotYear}a illustrates the yearly evolution of accidents. From 2013 to 2017, an increase of approximately a 50\% is observed. Figure~\ref{PlotYear}b shows the probability of accidents resulting in severe injuries. The values remain within the expected variation\footnote{The error bars are determined by assuming a binary (Bernoulli), distribution  between light and severe injuries. For the probability of a severe injury at a specific year, $p$, the data-based value is used as a good approximation to the true value. Therefore, the standard deviation of the mean has the following form:  $\sqrt{p(1-p)/N}$, where $N$ is the total number of accidents for each year.}.

\vspace{-0.5cm}

\begin{figure}[!htb]
\centering
\subfloat[Number of accidents.]{\includegraphics[width=0.23\linewidth]{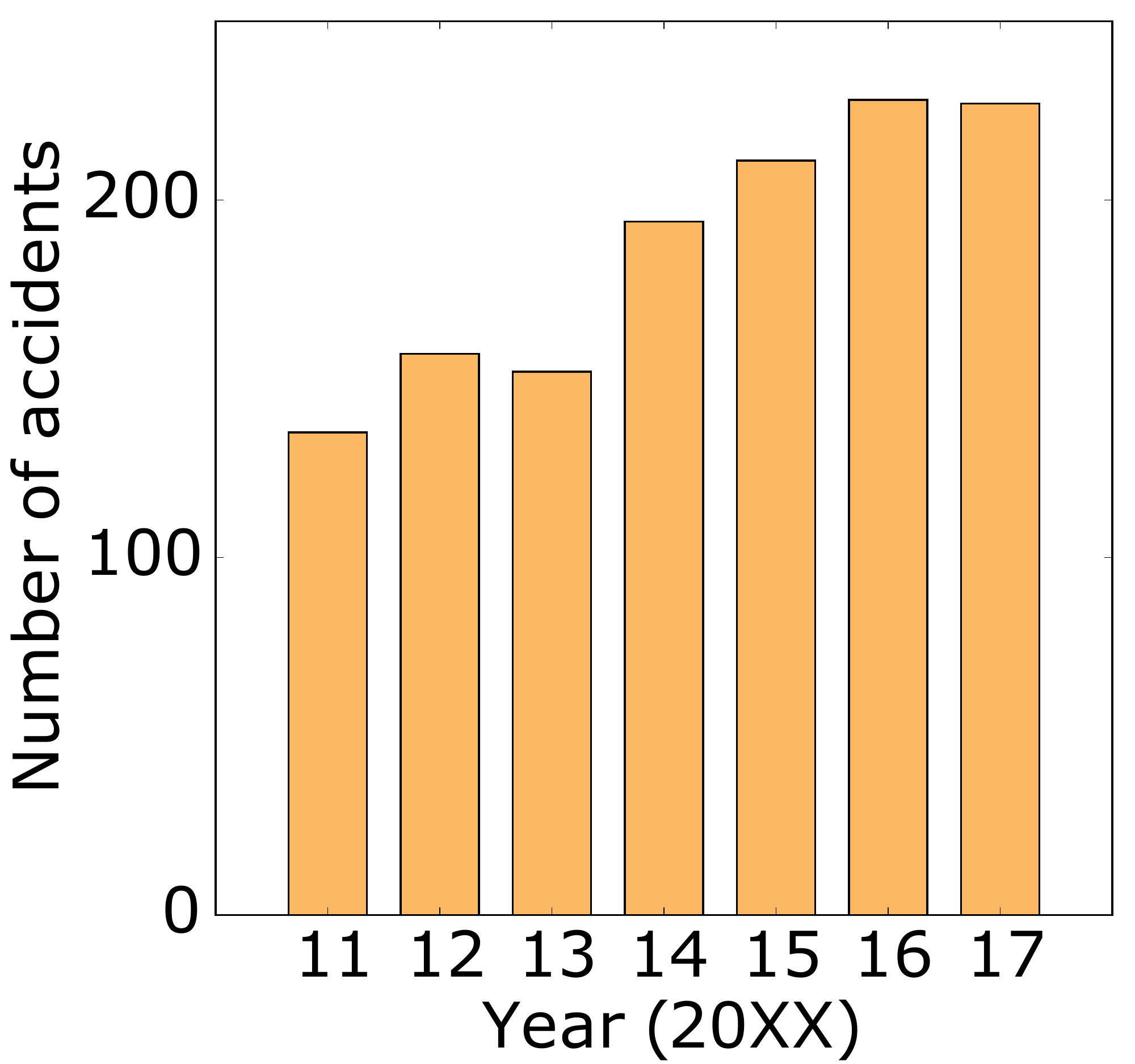}}\label{PlotAccidentsYear}
\subfloat[Severity of accidents.]{\includegraphics[width=0.23\linewidth]{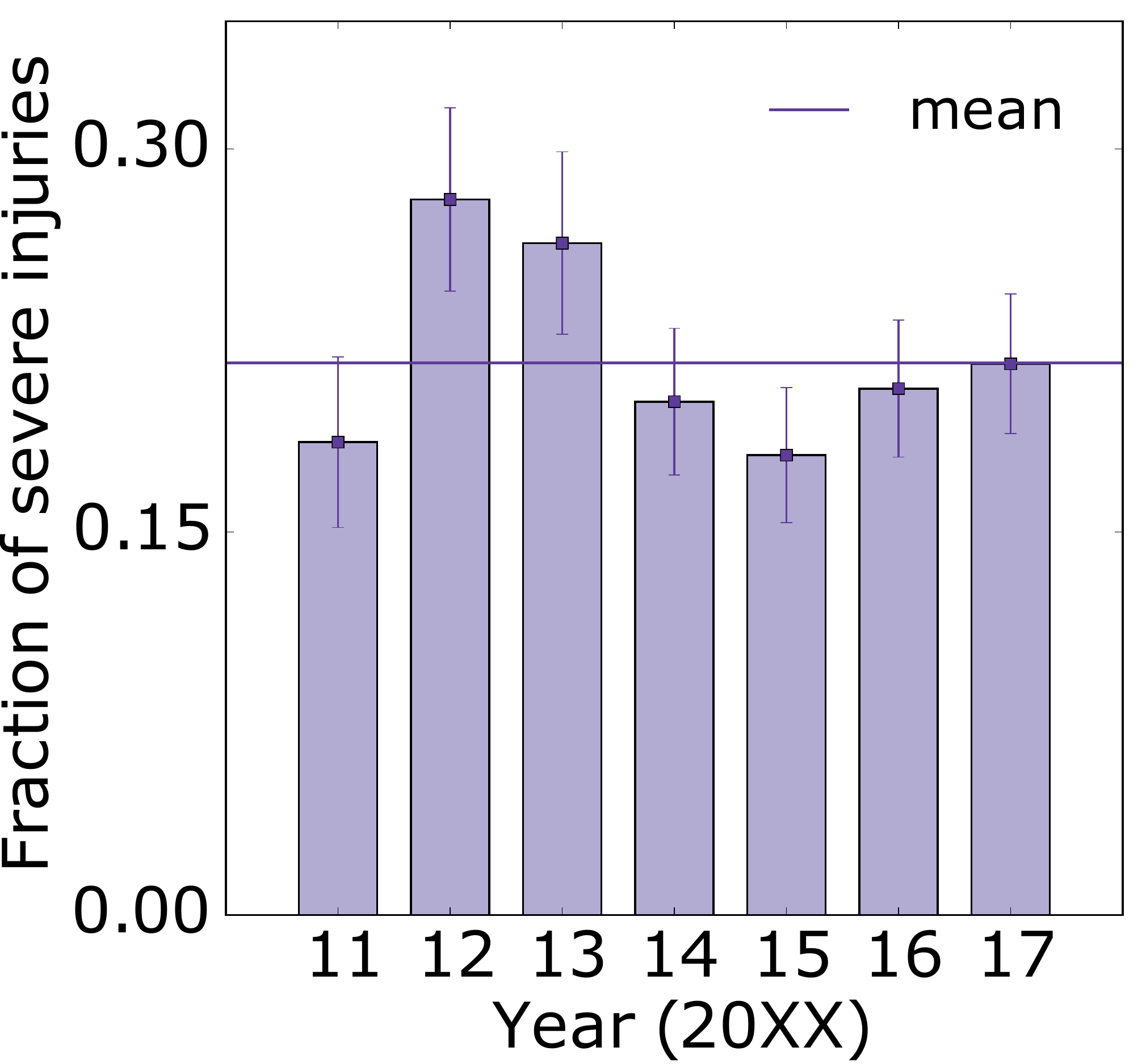}}\label{PlotPercentageSevereYear}
\subfloat[Accidents vs. temperature.]{\includegraphics[width=0.252\linewidth]{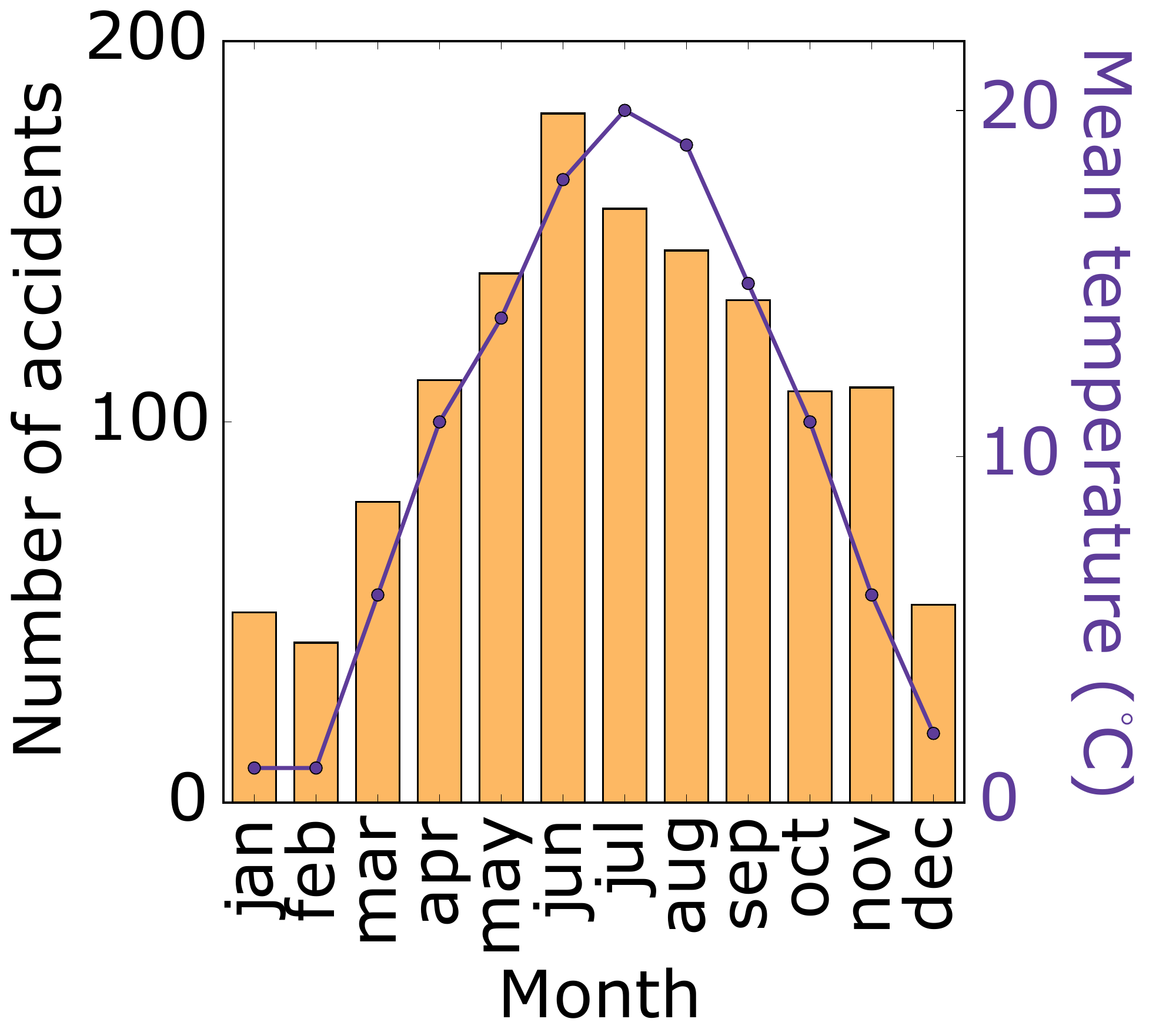}}\label{PlotAccidentsMonth}
\subfloat[Severity vs. precipitation.]{\includegraphics[width=0.252\linewidth]{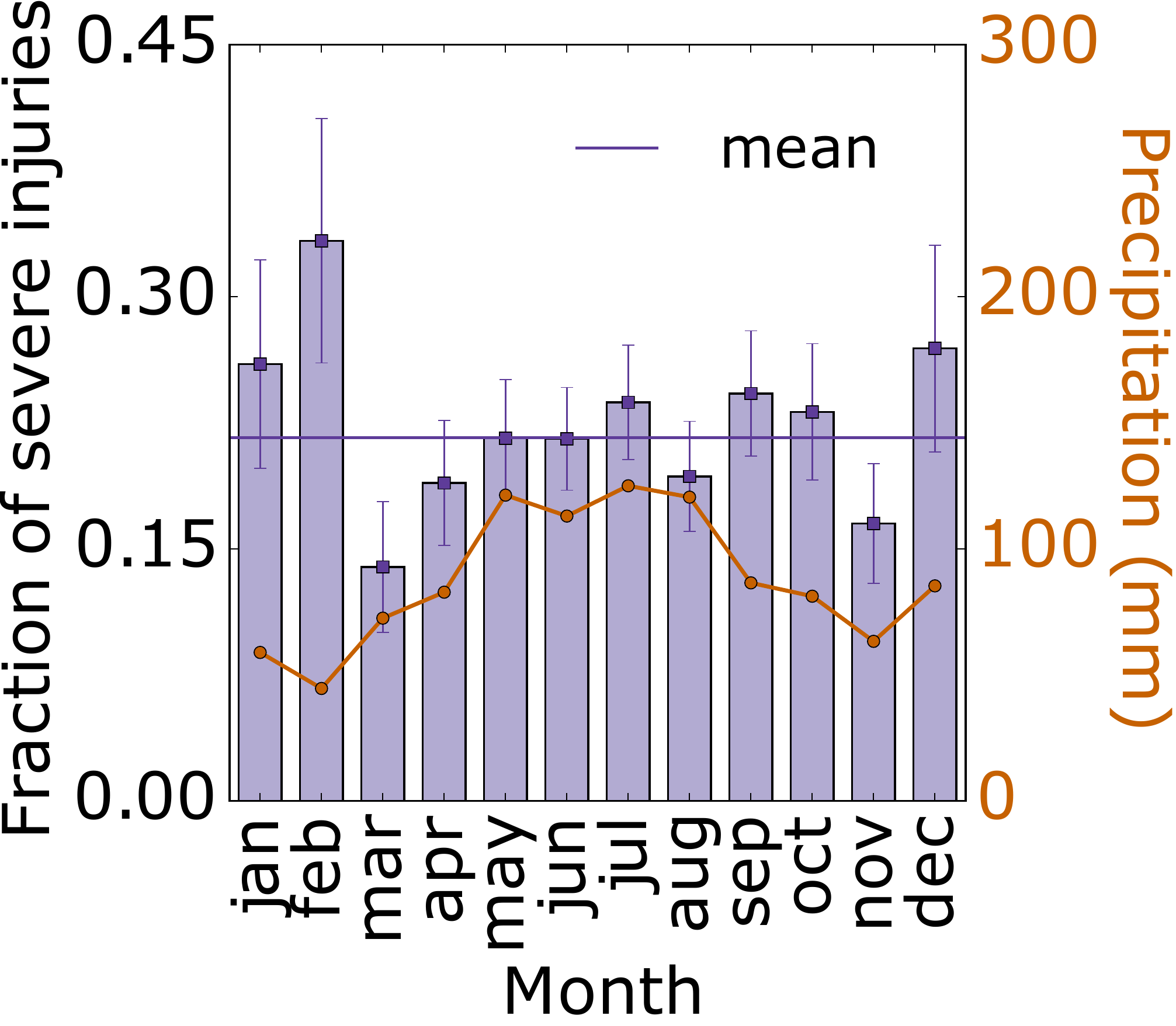}}\label{PlotPercentageAccidentsMonth}
\caption{Yearly and monthly analysis of accident and weather data.}\label{PlotYear}
\end{figure}

\vspace{-0.5cm}

The total number of accidents per month across all years vs. the average temperature per month\footnote{Available at https://www.timeanddate.com/weather/switzerland/zurich/climate (last access: May 2019).} is shown in Figure~\ref{PlotYear}c. That fewer accidents occur in colder months reflects the temporal transit patters of citizens in Z\"urich: choosing public transport or driving by car during the colder months of the year to avoid the discomfort of cycling in the cold. On the other hand, the severity of the accidents in relative number is higher during the winter months, as shown in Figure~\ref{PlotYear}d. Given the highly diverse grade, i.e. slope, of the studied region, snow and frozen street surfaces in winter are likely to explain this observation. Precipitation also seems important. Summer shows on average up to 70\% higher precipitation~\footnotemark[20] than the months of March and November during which the lowest fractions of severe injuries are observed. 

Figure~\ref{PlotType}a illustrates the relation between the accidents and their causes. Self-caused accidents are predominant, a 40\% of of all accidents are self-caused. Additionally, they show the highest severity, 28\%, as shown in Figure~\ref{PlotType}b. Head-on collisions and accidents on crossing lanes follow in severity, suggesting that intersections entail higher risk overall. Given the predominance of self-caused accidents, the potential to improve bike riding safety via warnings and route recommendations is apparent. Such risk communication can improve cyclist awareness and simultaneously contribute to greater cycling confidence in tourists and new cyclists. 

\begin{figure}[h!]
\centering
\subfloat[Accidents per type.]{\includegraphics[width=0.235\linewidth]{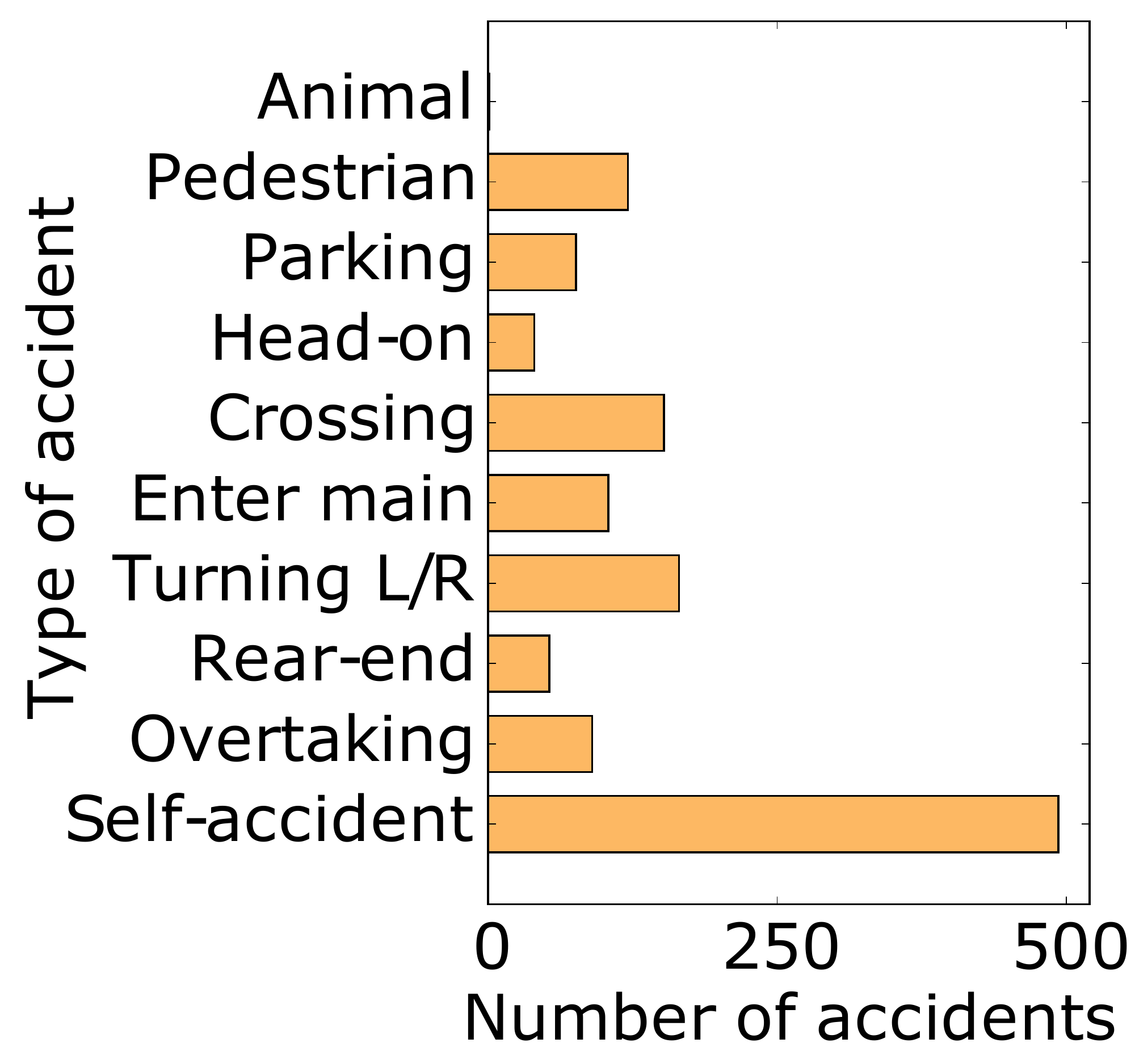}\label{PlotAccidentsType}
\includegraphics[width=0.235\linewidth]{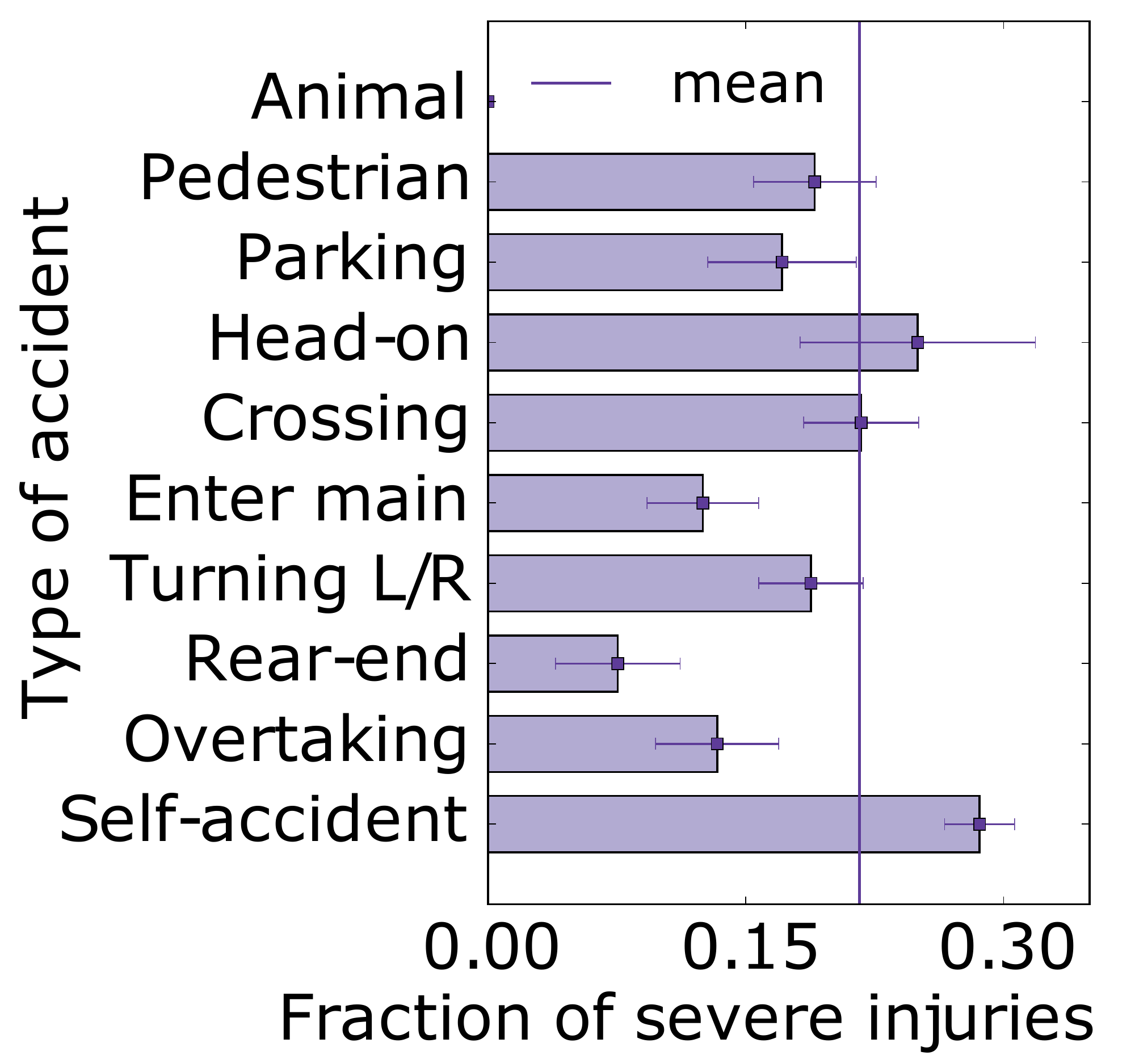}}\label{PlotPercentageSevereType}
\subfloat[Accidents over time and day of the week.]{\includegraphics[width=0.256\linewidth]{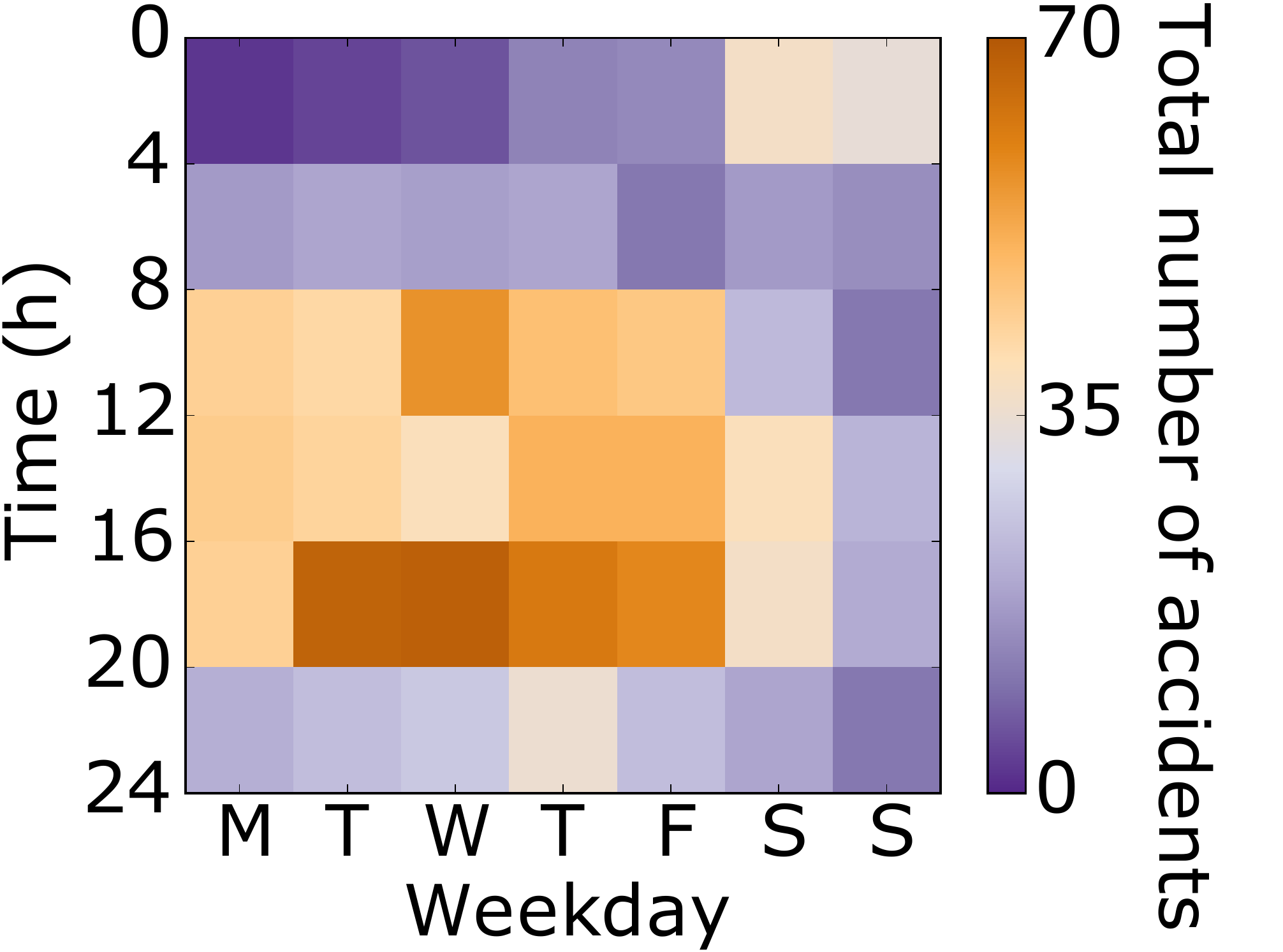}\label{PlotAccidentsHourWeekday}\includegraphics[width=0.256\linewidth]{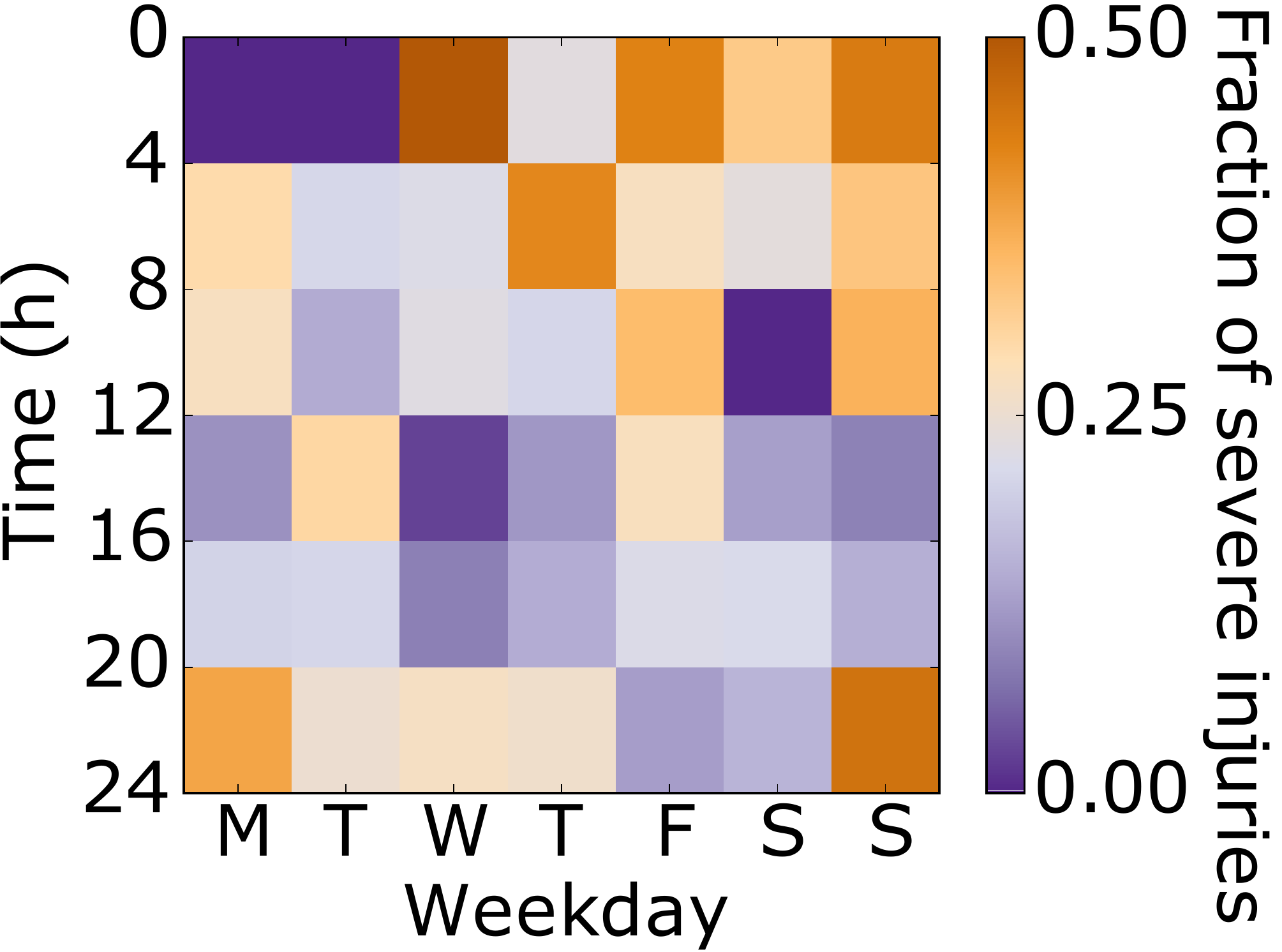}}
\caption{Type of accidents and their time of occurrence.}\label{PlotType}
\end{figure}

\vspace{-0.5cm}

Figure~\ref{PlotType}c illustrates the time of accident occurrences during weekdays, which show significantly more accidents than weekends. Weekday accidents usually happen early morning and late afternoon, suggesting that accidents happen during commuting times, i.e. home-work and vice versa. During weekends accidents mainly appear during the following times: (i) Saturday afternoon, probably corresponding to shopping/outings. (ii) Early morning hours on Saturday and Sunday, suggesting accidents related to poor visibility conditions, fatigue, and alcohol consumption. Figure~\ref{PlotType}d shows that the latter are the most severe ones.

\subsection{Bike route recommendations: safety vs. discomfort}\label{subsec:recommendations}

Compared to the bike routes of Google Maps, the recommended routes of the designed software artifact are similar in overall for $\alpha=0$ in Equation~\ref{comfortFnc}. Nevertheless, they are in general slightly shorter in distance and longer in time than those of Google Maps. This difference can be attributed to additional information potentially utilized by Google Maps, for instance information about the street infrastructure, drivers' and cyclists' route preference as well as traffic lights. 

Figure~\ref{routeComparer} shows the relative improvement of risk and discomfort between the 24 baseline routes and the recommended routes for different $\alpha$ values. A clear trade-off between safety and comfort is observed. Optimal values of $\alpha$ lie between 0.2 and 0.4.

%\vspace{-0.5cm}

\begin{figure}[!htb]
\centering
\includegraphics[width=0.33\linewidth]{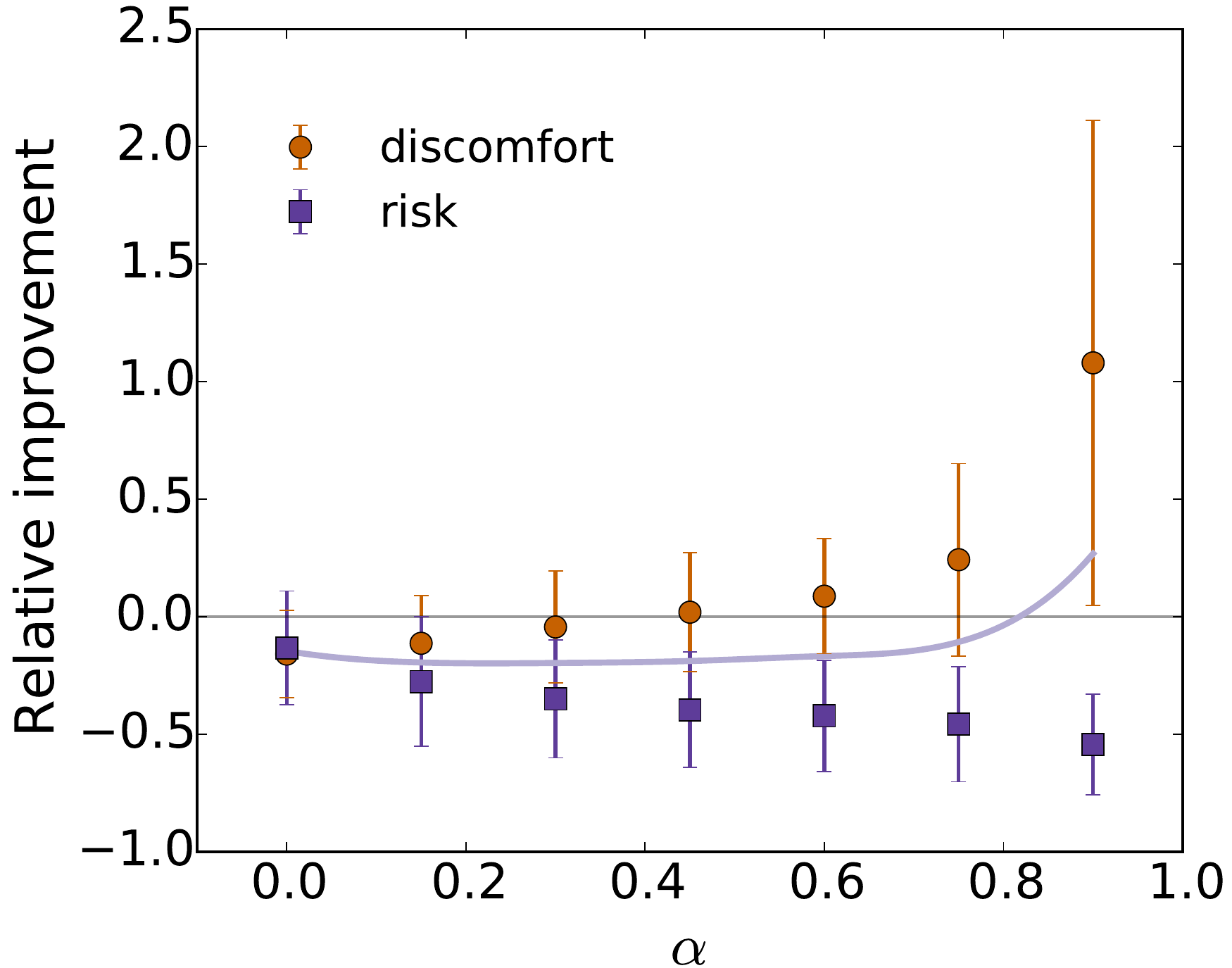}
\caption{Average relative improvement of risk and discomfort between baseline and recommended routes. The light purple line is the mean improvement between the two estimates. It can be used to assess $\alpha$ values with a good balance between the two.}\label{routeComparer}
\end{figure}

%\vspace{-0.5cm}

Figure~\ref{PlotFreqComp}b and~\ref{PlotFreqComp}c illustrate the changes in street utilization by increasing from $\alpha=0$ to $\alpha=0.5$ and $\alpha=0.75$ respectively, i.e. higher priority on safety is given to the recommended routes by BFS. The 2000 randomly generated departure and destination points are used for the mapping of street utilizations. 

\begin{figure}[!htb]
\centering
\subfloat[Recommended route on the interactive map.]{\includegraphics[width=0.302\linewidth]{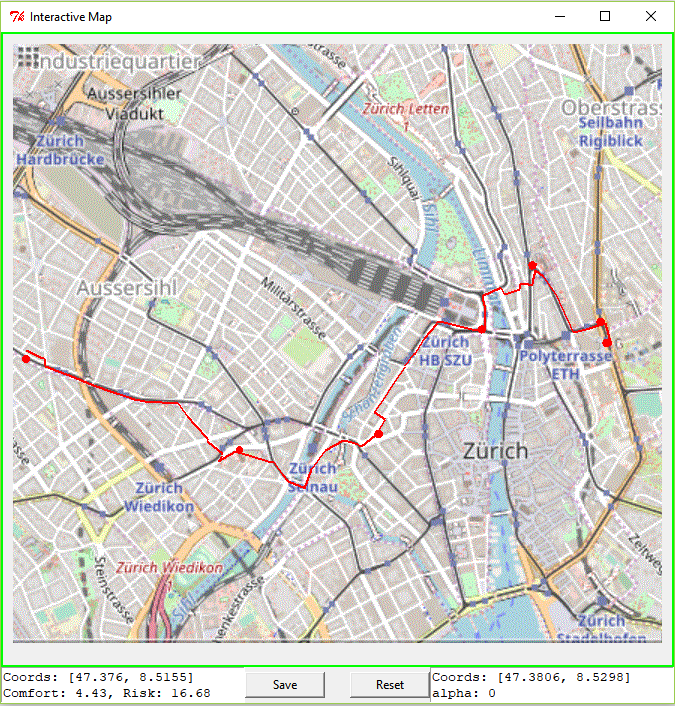}}
\subfloat[From $\alpha=0$ to $\alpha=0.5$.]{\includegraphics[width=0.345\linewidth]{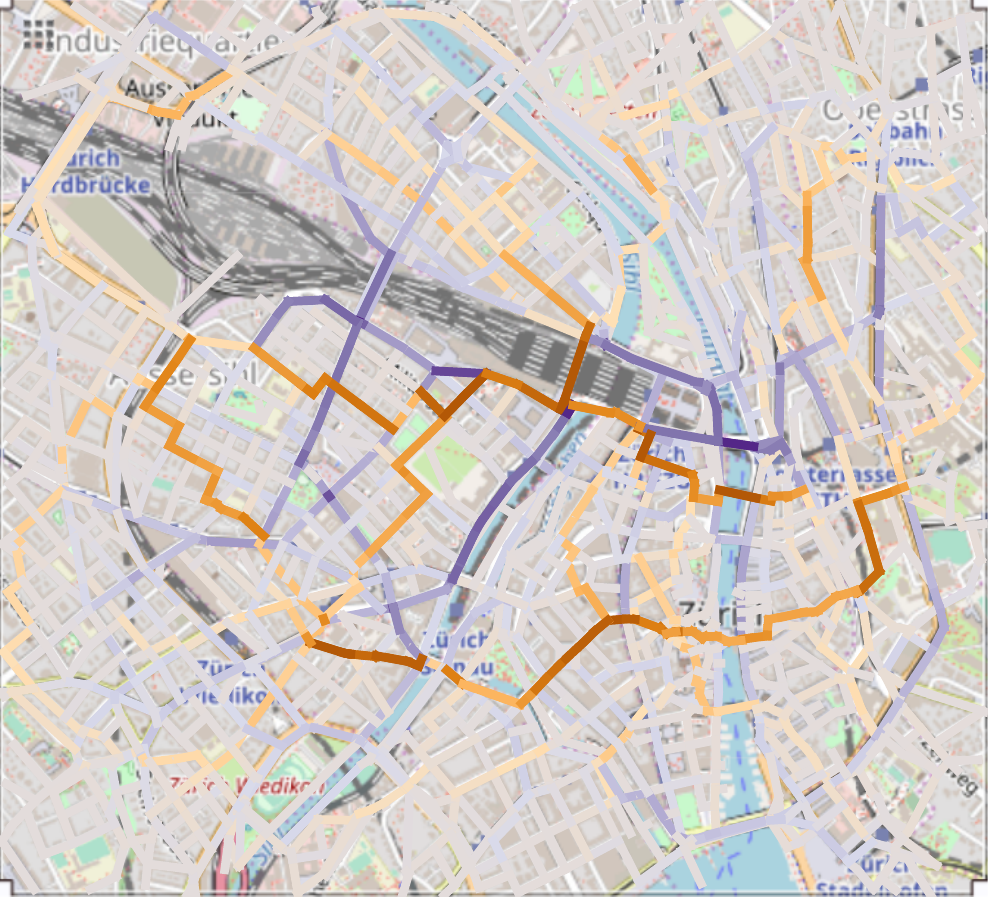}}\label{PlotFreqComp5}
\subfloat[From $\alpha=0$ to $\alpha=0.75$.]{\includegraphics[width=0.345\linewidth]{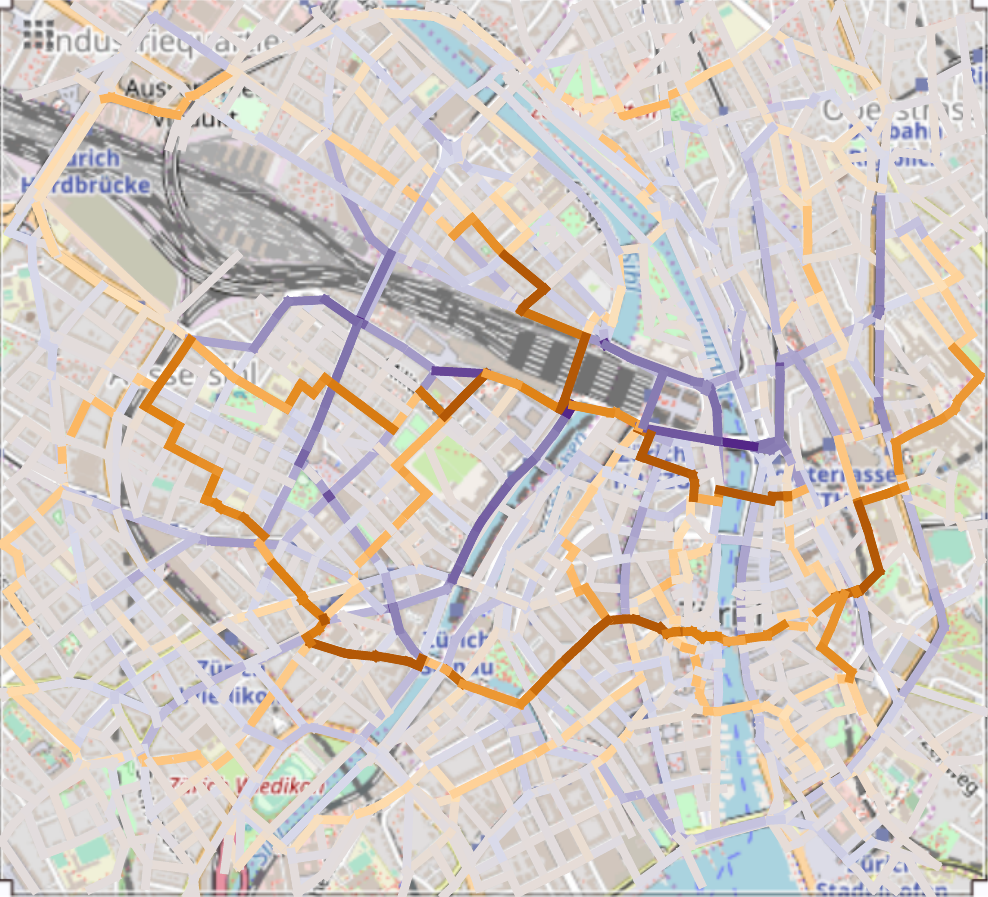}}\label{PlotFreqComp75}
\caption{The graphical user interface for interactive route recommendations and the changes in the frequency of street utilization by the recommended routes when changing from $\alpha=0$ to $\alpha=0.5$ and $\alpha=0.75$ respectively. Orange street segments indicate an increase in street utilization, while purple ones a decrease.}\label{PlotFreqComp}
\end{figure}

%\vspace{-0.5cm}

The colored maps show that areas such as Langstrasse\footnotemark[17] and Bahnhofstrasse\footnote{Available at https://en.wikipedia.org/wiki/Bahnhofstrasse (last accessed: May 2019).} are avoided already for $\alpha = 0.5$, while two cross-city routes become dominant.

%\section{Comparison with Related Work}\label{sec:related-work}

\section{Conclusion and Future Work}\label{sec:conclusion}

This paper concludes that a data-driven approach for the estimation and mapping of cycling risk in complex evolving urban environments can provide invaluable empirical insights about safety. This is shown for the city center of Z\"urich, in which continuous risk contours are calculated based on historical geolocated accident data and information about their severity, linked to compensation policies of health insurances. Findings shows that bike accidents increase at a higher rate than bike use, while weather, seasonality, day of the week and time play a role on the likelihood of an accident and its severity. The predominance of self-caused accidents suggests the requirement for a higher awareness of risks and safe routing information. This requirement is met by personalized route recommendations that balance safety and comfort. The findings of this paper have an impact on the following: (i) The cyclists' risk awareness and safety improvement. (ii) Policy-making for improving transport infrastructure and encourage further the use of environmentally friendly transport means such as bikes by existing city habitants, new cyclists as well as tourists. 

Future work includes the expansion of the risk and discomfort estimation with exposure and vibration measures~\cite{USDEPTRANREP,Bil2015}, the study and comparison of several other cities\cite{Kaygisiz2017}, the influence of other traffic in cycling safety as well as the design of traffic simulation models for the participatory multi-objective optimization of traffic flows~\cite{Schmid2017,Amini2017,Pournaras2018}.

\begin{acknowledgements}
The authors would like to thank Leonie Fl\"uckiger for her support to this project. 
\end{acknowledgements}

%\vspace{-0.5cm}

\bibliographystyle{spmpsci}
\bibliography{bike-safety}

%% Non-BibTeX users please use
%\begin{thebibliography}{}
%%
%% and use \bibitem to create references. Consult the Instructions
%% for authors for reference list style.
%\end{thebibliography}

\end{document}